\documentclass{article}
\usepackage{amsmath}
\usepackage{amsthm}
\usepackage{amsfonts}
\usepackage{tikz}
\usepackage[margin=2.2cm]{geometry}
\usepackage{hyperref}

\newtheorem{theorem}{Theorem}[section]

\newtheorem{lemma}[theorem]{Lemma}

\newtheorem{corollary}[theorem]{Corollary}

\newtheorem{conj}[theorem]{Conjecture}

\newcommand{\R}{\mathbb{R}}
\newcommand{\C}{\mathbb{C}}
\newcommand{\E}{\mathbb{E}}
\newcommand{\be}{\begin{equation}}
\newcommand{\ee}{\end{equation}}

\newcommand{\Mat}{\operatorname{Mat}}

\title{A Note on Mixed Matrix Moments for the Complex Ginibre Ensemble}
\author{Meg Walters$^{1}$ and Shannon Starr$^{2}$\\
\small
${}^1$ University of Rochester, Department of Mathematics, Rochester, NY 14627\quad \url{meg.walters@rochester.edu}\\
\small
${}^2$ University of Alabama at Birmingham, Applied Mathematics, Birmingham, AL 35294--1170\quad \url{slstarr@uab.edu}
}
\date{21 November 2014}

\begin{document}

\maketitle
\begin{abstract}
\setcounter{section}{0}
We consider the mixed matrix moments for the complex Ginibre ensemble.
These are well-known.
We consider the relation to the expected overlap functions of Chalker and Mehlig.
This leads to new asymptotic problems for the overlap.
We obtain some results, but we also state some remaining open problems.
\end{abstract}

\section{Introduction}

The purpose of this note is to make some observations about the mixed matrix moments for non-Hermitian random matrices.
Let $\Mat_n(\C)$ denote the set of complex $n\times n$ matrices (since we will use the letter $M_n$ for something else).

The model we will focus on most is the complex Ginibre ensemble, 
\be
A_n \in \Mat_n(\C)\, ,\quad A_n = (a_n(j,k))_{j,k=1}^{n}\, ,\quad
a_n(j,k)\, =\, \frac{X(j,k)+iY(j,k)}{\sqrt{2n}}\, ,
\ee
where $(X(j,k))_{j,k=1}^{\infty}$, $(Y(j,k))_{j,k=1}^{\infty}$ are IID, $\mathcal{N}(0,1)$ real random variables.

Much of what we will say has already been explored by Chalker and Mehlig in a pair of papers \cite{CM,MC},
in particular in their definition of expected overlap functions.
There are other models of interest which were explored by Fyodorov and coauthors \cite{FY1,FY2},
for which one can obtain more explicit formulas for the expected overlap functions.
But our main emphasis will be to relate Chalker and Mehlig's formulas for the complex Ginibre
ensemble to the mixed matrix moments.

Our motivation in considering this problem is the following. There is a rough analogy between mean-field spin glasses and random matrices,
as far as the mathematical methods are concerned.
We indicate this in the table in Figure \ref{fig:1a}.
(See section 2 for more details and references.)
This leads to a method to calculate moments. But there is still the question about how to relate
the moments to the spectral information for the matrix.
\begin{figure}[h]
\begin{center}
\boxed{
\begin{minipage}{7cm}
\begin{center}
{\bf
Random Matrices}\\[5pt]
expectation of moments\\
recurrence relation for moments\\
formula for Stieltjes transform of limiting law
\end{center}
\end{minipage}
\hspace{1cm}
\begin{minipage}{9cm}
\begin{center}
{\bf Spin Glasses}\\[5pt]
expectation of products of overlaps\\
stochastic stability equations: Ghirlanda-Guerra identities\\
proof of Parisi's ultrametric ansatz
\end{center}
\end{minipage}
}
\caption{Some analogous elements in random matrix  and spin glass theory. (Proofs may differ considerably.) \label{fig:1a}}
\end{center}
\end{figure}

We are going to start, in Section 2, by briefly recalling the formula for the mixed matrix moments of the complex Ginibre ensemble, and we will emphasize the relation to spin glass techniques.
This is already known. 
We will give references.

Then, in Sections 3 through 8, we will describe how this may be calculated from the expected overlap functions of Chalker and Mehlig.
On the other hand this leads to new problems.
This is the main subject of this note.

\section{Mixed matrix moments}

Given any positive integer $k$ and any nonnegative integers $p(1),q(1),\dots,p(k),q(k)$, we may define
\be
M_n(\mathbf{p};\mathbf{q})\, 
=\, \frac{1}{n}\, \operatorname{tr}[A_n^{p(1)} (A_n^*)^{q(1)} \cdots A_n^{p(k)} (A_n^*)^{q(k)}]\, ,
\ee
for $\mathbf{p}=(p(1),\dots,p(k))$, $\mathbf{q}=(q(1),\dots,q(k))$. Also, $M_0=1$.
An example is 
\be
M_n((2,2);(2,2))\, =\, \frac{1}{n}\, \sum_{j_1,\dots,j_8=1}^n a_n(j_1,j_2) a_n(j_2,j_3) \overline{a}_n(j_4,j_3) \overline{a}_n(j_5,j_4) a_n(j_5,j_6) a_n(j_6,j_7)
\overline{a}_n(j_8,j_7) \overline{a}_n(j_1,j_8)\, .
\ee
Since $a_n(j,k) = (X(j,k)+iY(j,k))/\sqrt{2n}$, Wick's rule (or Gaussian-integration-by-parts)
gives the formula 
\be
\E[a_n(j,k) a_n(j',k')]\, =\, \E[\overline{a}_n(j,k) \overline{a}_n(j',k')], =\, 0\ \text{ and } \
\E[a_n(j,k) \overline{a}_n(j',k')]\, =\, n^{-1} \delta_{j,j'} \delta_{k,k'}\, .
\ee
Using this, and defining $m_n(\mathbf{p},\mathbf{q}) = \E[M_n(\mathbf{p},\mathbf{q})]$, Gaussian integration-by-parts (or Wick's rule) leads to:
\begin{equation}
\label{eq:recurrence}
m_n(\mathbf{p},\mathbf{q})\, =\, \sum_{(\mathbf{p}',\mathbf{q}',\mathbf{p}'',\mathbf{q}'') \in \mathcal{S}(\mathbf{p},\mathbf{q})} 
%\alpha\left(\begin{matrix}\mathbf{p},\mathbf{q}\\ \mathbf{p}',\mathbf{q'},\mathbf{p}'',\mathbf{q}''\end{matrix}\right) 
\E[M_n(\mathbf{p}',\mathbf{q}') M_n(\mathbf{p}'',\mathbf{q}'')]\, ,
\end{equation}
where $\mathcal{S}(\mathbf{p},\mathbf{q})$ is the set of all admissible pairs, which we describe now.
Let $R = p(1)+\dots+p(k)+q(1)+\dots+q(k)$, 
and define $\sigma = (\sigma(1),\dots,\sigma(R)) \in \{+1,-1\}^R$
as $\sigma\, =\, ((+1)^{p(1)},(-1)^{q(1)},\dots,(+1)^{p(k)},(-1)^{q(k)})$
viewed as spins on vertices arranged on a circle.
We will sometimes denote this as $\sigma_{\mathbf{p},\mathbf{q}}$.
Let $\Sigma(\mathbf{p},\mathbf{q})$ denote pairs $(\sigma',\sigma'')$ as follows.
We match up the first $+1$ and any $-1$.
Where these two are removed, we pinch the circle into two smaller circles. 
Then the remaining spins on the two smaller circles comprise $\sigma'$ and $\sigma''$.
E.g., for a particular example
\be
\sigma = (\underline{+1},+1,\underline{-1},-1,+1,+1,-1,-1) \mapsto (\sigma',\sigma'') = ((+1),(-1,+1,+1,-1,-1))\, .
\ee
The set $\Sigma(\mathbf{p},\mathbf{q})$ is the set of all possible pairs $(\sigma',\sigma'')$ obtainable in this way.
We then define $\mathcal{S}(\mathbf{p},\mathbf{q})$ to be the set of all pairs $(\mathbf{p}',\mathbf{q}')$ and $(\mathbf{p}'',\mathbf{q}'')$ by mapping
backwards $\Sigma(\mathbf{p},\mathbf{q})$ from $\sigma'$ and $\sigma''$, this way.

Using this, we wish to indicate the proof of the following theorem.
\begin{theorem}
\label{thm:MMM}
For any $k$ and any $\mathbf{p},\mathbf{q}$, we have
$$
\lim_{n \to \infty} m_n(\mathbf{p},\mathbf{q})\, =\, m(\mathbf{p},\mathbf{q})\, ,
$$
where $m(\mathbf{p},\mathbf{q})$ is as follows.
Let $C_R$ denote the number of all non-crossing matchings of
$R$ vertices on a circle (Catalan's number).
Let $m(\mathbf{p},\mathbf{q})$ denote the cardinality of all such matchings
satisfying the following constraint: assigning spins to the $R$ vertices by 
 $\sigma_{\mathbf{p},\mathbf{q}}$,  
each edge has two endpoints with one $+1$ spin and one $-1$ spin.
\end{theorem}
As an example,  
$m((2,2);(2,2)) = 3$ where the matchings are indicated diagrammatically in \ref{fig:1}.
\begin{figure}
\begin{center}
\includegraphics[width=4cm,height=3.2cm]{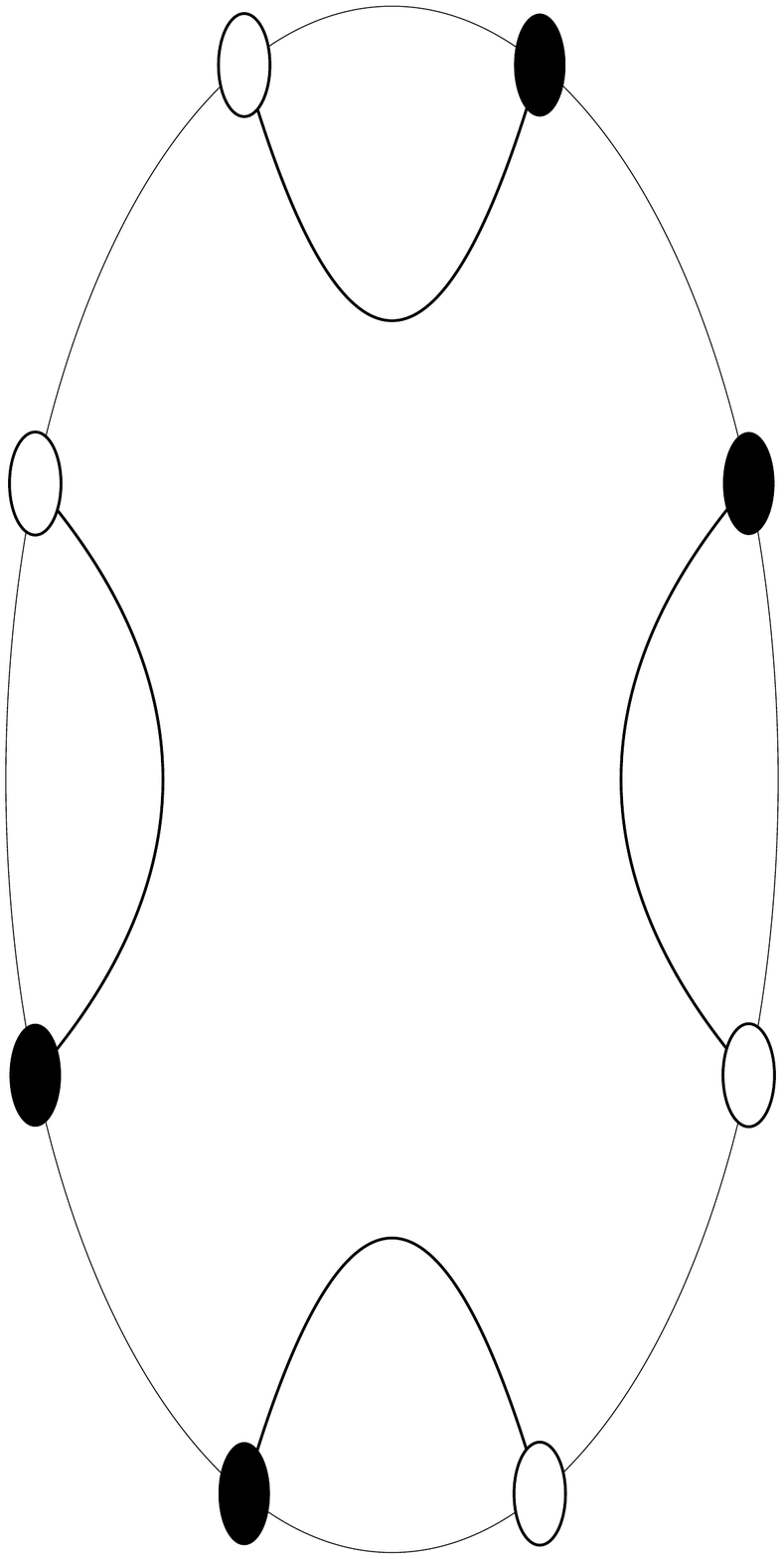}\qquad
\includegraphics[width=4cm,height=3.2cm]{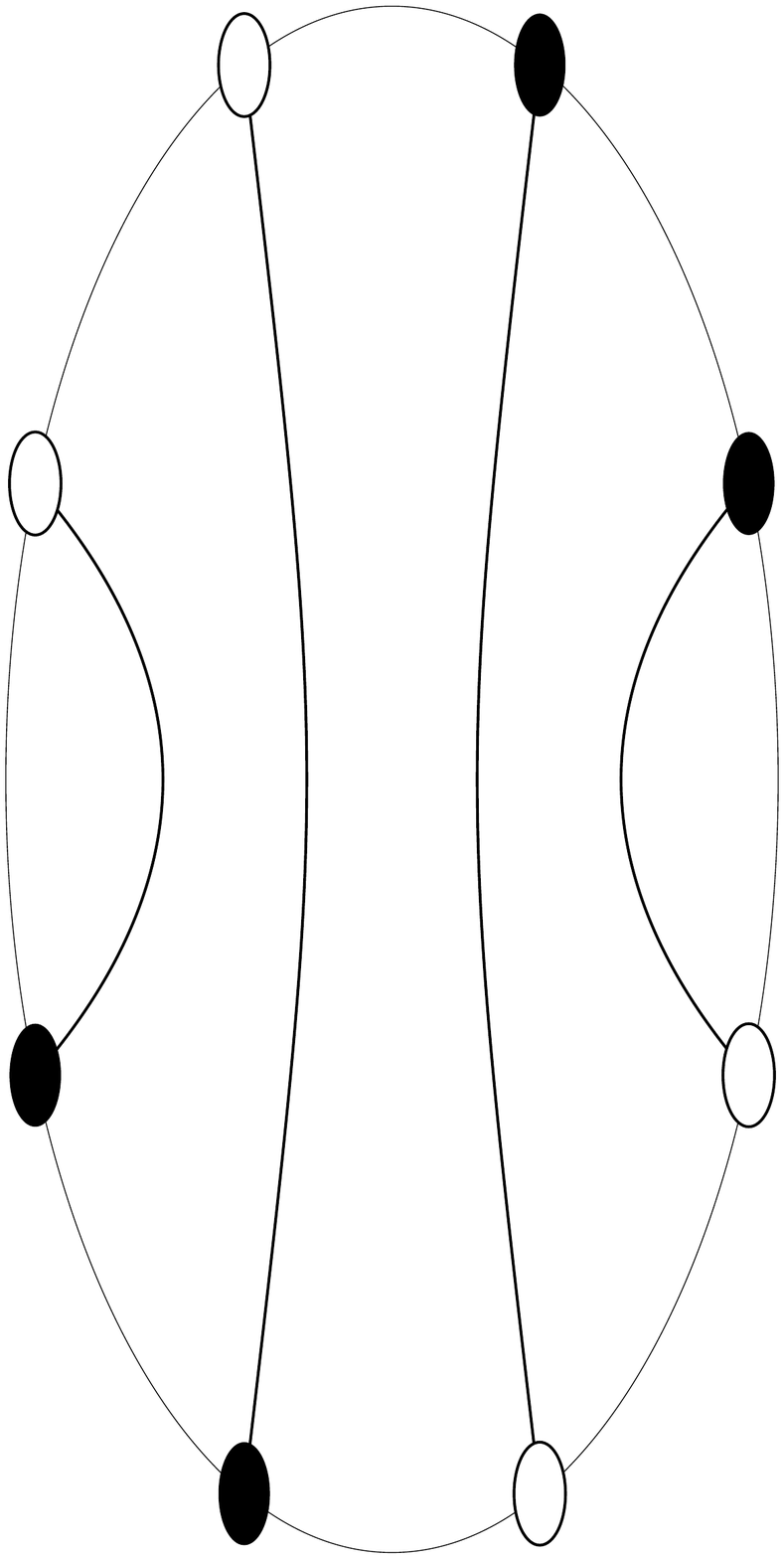}\qquad
\includegraphics[width=4cm,height=3.2cm]{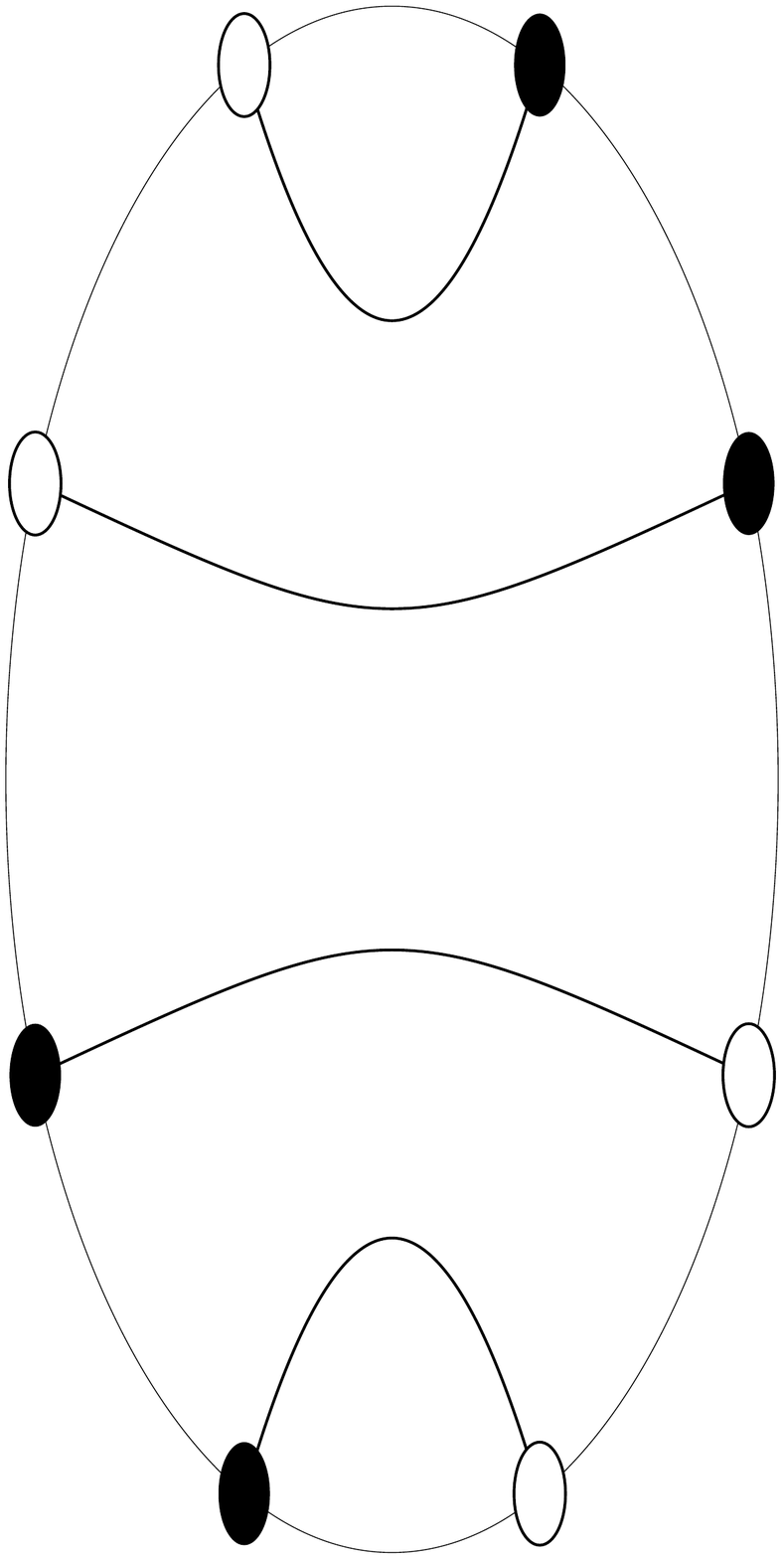}
\end{center}
\caption{
A figure for the matchings in $m((2,2);(2,2))$.
\label{fig:1}
}
\end{figure}
Theorem \ref{thm:MMM} is a well-known result. We refer to \cite{Kemp} for a discussion.
We will motivate a proof of this result, without including all details, here.
Our reason is that we actually want to use this result to motivate the discussion of random matrices and spin glasses further, which we indicated in Section 1.

\subsection{Argument for the proof of the mixed matrix moments}
\label{subsec:Argument}
The first step in the argument for the proof of Theorem \ref{thm:MMM} is to use concentration of measure (COM) to replace (\ref{eq:recurrence})
with a nonlinear recurrence relation.
Here what we mean is non-linearity in the probability measure for the random entries of the matrix.
Since the expectation is linear, what we really mean is to obtain a product of two expectations.
If $M_n(\mathbf{p}',\mathbf{q}')$ and $M_n(\mathbf{p}'',\mathbf{q}'')$ were independent, then we could
replace the expectation by a product.
But they are not exactly independent.
Instead, they satisfy COM, which means that they are approximately non-random.
And, of course, non-random variables are exactly independent of every other random variable
(as well as themselves).

The easiest version of COM is just $L^2$-concentration. For example, the following lemma is very easy to prove:
\begin{lemma}
\label{lem:unifGrad}
Suppose $f : \R^n \to \R$ is a function such that 
$\|\nabla f\|^2_{\infty} = \sup_{\mathbf{x} \in \R^n} \sum_{k=1}^n \left(\frac{\partial f}{\partial x_k}(\mathbf{x})\right)^2$ is finite.
Then if $U_1,\dots,U_n,V_1,\dots,V_n$ are IID $\mathcal{N}(0,1)$ random variables then
\be
\label{eq:lemG}
\E\left[\left(f(\mathbf{U})-f(\mathbf{V})\right)^2\right]\, \leq\, 2 \|\nabla f\|_{\infty}^2\, .
\ee
\end{lemma}
This can be proved using the basic, but important method of ``quadratic interpolation,'' which is sometimes
called the ``smart path method'' by some mathematicians working on spin glasses.
\begin{proof}
Let $\boldsymbol{Z} = (Z_1,\dots,Z_n)$ be an IID $\mathcal{N}(0,1)$ vector, independent of 
$\mathbf{U}$ and $\mathbf{V}$. Then define $\widetilde{\mathbf{U}}(\theta) = \sin(\theta)\, \mathbf{U} 
+ \cos(\theta)\, \mathbf{Z}$ and $\widetilde{\mathbf{V}}(\theta) = \sin(\theta)\, \mathbf{V} 
+ \cos(\theta)\, \mathbf{Z}$.
Then $\frac{d}{d\theta} \widetilde{\mathbf{U}}(\theta) = \widetilde{\mathbf{U}}(\theta+\frac{\pi}{2})$,
and $\E[\widetilde{\mathbf{U}}(\theta) \widetilde{\mathbf{U}}(\theta+\frac{\pi}{2})] = 0$.
This means that $\widetilde{\mathbf{U}}(\theta)$ is statistically independent of its $\theta$-derivative.
Similar results hold for $\widetilde{\mathbf{V}}(\theta)$.
On the other hand $\E[\widetilde{\mathbf{V}}(\theta) \widetilde{\mathbf{U}}(\theta+\frac{\pi}{2})] = 
-\sin(\theta)\cos(\theta)$.

Next, using the fundamental theorem of calculus,
\be
\E\left[\left(f(\mathbf{U})-f(\mathbf{V})\right)^2\right]\,
=\, \int_0^{\pi/2} \frac{d}{d\theta}
\E\left[\left(f(\widetilde{\mathbf{U}}(\theta))-f(\widetilde{\mathbf{V}}(\theta))\right)^2\right]\,
d\theta\, ,
\ee
and an easy calculation using Wick's rule (and the covariance formulas mentioned above) shows that 
\be
\frac{d}{d\theta}
\E\left[\left(f(\widetilde{\mathbf{U}}(\theta))-f(\widetilde{\mathbf{V}}(\theta))\right)^2\right]\,
=\, 2\sin(2\theta) \E\left[\nabla f(\widetilde{\mathbf{U}}(\theta)) \cdot 
\nabla f(\widetilde{\mathbf{V}}(\theta))\right]\, .
\ee
Then (\ref{eq:lemG}) follows by using the Cauchy-Schwarz inequality.
\end{proof}
This is only the simplest Gaussian COM result. More powerful bounds are also attainable by the same methods.
We mentioned this one because the proof is the simplest. Others are not harder, but build on this idea.

This lemma is just a tool. One must still apply it to the case at hand to prove that the various mixed matrix moments $M_n(\mathbf{p},\mathbf{q})$
do satisfy COM.
%We will give the details for this problem in the appendix.
It is an interesting calculation, and much of the combinatorics, especially involving matchings related to Catalan's number,
are first visible in the grad-squared calculation.
For now, however, let us state that the desired COM result is true.

Then we are able to boost (\ref{eq:recurrence}) to 
\begin{equation}
\label{eq:recur2}
\lim_{n \to \infty} m_n(\mathbf{p},\mathbf{q}) - 
\sum_{(\mathbf{p}',\mathbf{q}',\mathbf{p}'',\mathbf{q}'') \in \mathcal{S}(\mathbf{p},\mathbf{q})} 
m_n(\mathbf{p}',\mathbf{q}') m_n(\mathbf{p}'',\mathbf{q}'')\, =\, 0\, .
\end{equation}
Another easy fact is that, due to symmetry, $m_n(\mathbf{p},\mathbf{q})=0$ unless $p(1)+\dots+p(k)=q(1)+\dots+q(k)$.
And, of course, $m_0=1$.

Using this and the method of induction it is easy to prove Theorem \ref{thm:MMM}.
%There is a combinatorial aspect to this, too.
%We will also discuss this in the appendix.

\subsection{Commentary on proof technique}

The quadratic interpolation technique is important in spin glasses. 
The first major use was by Guerra and Toninelli \cite{GT} and Guerra \cite{G}.
It is called the ``smart path method'' by Talagrand \cite{Talagrand}.
Talagrand gives a COM bound using this method with Gaussian tails in Section 1.3 of his book.

Using Wick's rule to obtain a recurrence relation is important in many subjects.
It is a standard approach to determining moments of random matrices.
See, for instance, \cite{AndersonGuionnetZeitouni}, Chapter 1.
In the context of Gaussian spin glasses, this technique combined with stochastic stability
leads to the Aizenman-Contucci identities \cite{AC}.
When combined with COM it leads to the Ghirlanda-Guerra identities \cite{GG}.
See, for instance, the review \cite{ContucciGiardina}.

For random matrices, the problem of recombining the moments into useful information
about the limiting empirical spectral measure is also important.
For Hermitian random matrices, this is related to the classical moment method.
The standard approach is to put the moments together into the Stieltjes transform,
and then to proceed from there \cite{Pastur}.
Again, a good general reference is \cite{AndersonGuionnetZeitouni}, Chapter 1.

For spin glasses, the problem of integrating the Ghirlanda-Guerra identities
into a useful result for mean-field models was solved only relatively recently.
Panchenko showed that the ``extended Ghirlanda-Guerra identities'' imply
Parisi's ultrametric ansatz \cite{Panchenko}.
This is an important work. One element of his proof is putting various terms together into a
a new exponential type generating function. 
This might be somewhat analogous to the Stieltjes transform step.
But after that, the proofs are very different.

For non-Hermitian random matrices, getting useful information from the moments
is the topic we focus on, next.

\section{The expected overlap functions}
\label{sec:Expected}

For the moments $M_n(\mathbf{p},\mathbf{q})$, since they satisfy COM, one is primarily only interested in their expectations.
The next quantity we introduce is also defined just for the expectation.
(Studying its distribution may be interesting, but we will not comment on this, here.)
It is the expected overlap function of Chalker and Mehlig, introduced in \cite{CM} and further studied by them in \cite{MC}.

Given $A_n \in \Mat_n(\C)$, randomly distributed according to Ginibre's ensemble, almost surely it may be diagonalized.
This means that we can find eigenvalues $\lambda_1,\dots,\lambda_n \in \C$ as well as pairs of vectors $\psi_1,\phi_1,\dots,\psi_n,\phi_n \in \C^n$ such that
\be
A_n \psi_k\,=\, \lambda_k \psi_k\, ,\quad
\phi_k^* A_n\, =\, \lambda_k \phi_k^*\, ,\quad
\phi_k^* \psi_j\, =\, \delta_{jk}\, .
\ee
Using this, for any other vector $\Psi \in \C^n$, there is the formula
\be
A_n \Psi\, =\, \sum_{k=1}^{n} \lambda_k \langle \phi_k, \Psi \rangle \psi_k\, .
\ee
These are random because they depend on $A_n$. But we may take the expectation over the randomness.

Given any continuous function, $f$ with compact support on $\C$, one may define
\be
\omega^{(1)}_n[f]\, =\, \frac{1}{n}\, \E\left[\sum_{k=1}^{n} f(\lambda_k) \|\phi_k\|^2 \|\psi_k\|^2\right]\, .
\ee
Similarly, given any continuous function $F$, with compact support on $\C \times \C$, we may define
\be
\omega^{(2)}_n[F]\, =\, \frac{1}{n}\, \E\left[\sum_{j=1}^{n} \sum_{k\neq j} F(\lambda_j,\lambda_k) \langle \psi_k,\psi_j \rangle \langle \phi_j, \phi_k\rangle\right]\, .
\ee
Regularity of the eigenvalues and eigenvectors with respect to the matrix entries guarantee existence of functions $\mathcal{O}_n^{(1)} : \C \to \C$ and
$\mathcal{O}_n^{(2)} : \C \times \C \to \C$ such that
\be
\omega^{(1)}_n[f]\, =\, \int_{\C} f(z) \mathcal{O}_n^{(1)}(z)\, d^2z\quad \text{and}\quad
\omega^{(2)}_n[F]\, =\, \int_{\C}\int_{\C} F(z,w) \mathcal{O}_n^{(2)}(z,w)\, d^2z\, d^2w\, .
\ee
Using these definitions, one may determine a relation between the expected overlap functions
and the correlation functions for the eigenvalues.
Define $\rho^{(1)}_n$ and $\rho^{(2)}_n$, analogously to $\omega^{(1)}_n$ and $\omega^{(2)}_n$
as
\be
\rho^{(1)}_n[f]\, =\, \frac{1}{n}\, \E\left[\sum_{k=1}^{n} f(\lambda_k)\right]\, ,\
\text{ and }\
\rho^{(2)}_n[F]\, =\, \frac{1}{n}\, \E\left[\sum_{j=1}^{n} \sum_{k\neq j} F(\lambda_j,\lambda_k)\right]\, .
\ee
Then there are functions $\mathcal{R}_n^{(1)} : \C \to \C$ and $\mathcal{R}_n^{(2)} : \C \times \C \to \C$ such that
\be
\rho^{(1)}_n[f]\, =\, \int_{\C} f(z) \mathcal{R}_n^{(1)}(z)\, d^2z\quad \text{and}\quad
\rho^{(2)}_n[F]\, =\, \int_{\C}\int_{\C} F(z,w) \mathcal{R}_n^{(2)}(z,w)\, d^2z\, d^2w\, .
\ee
Then
\be
\mathcal{O}_n^{(1)}(z) + \int_{\C} \mathcal{O}_n^{(2)}(z,w)\, d^2w\, =\, \mathcal{R}_n^{(1)}(z)\, .
\ee
In terms of these functions, for any nonnegative integers $p$ and $q$,
\be
\label{eq:mnpq}
m_n((p);(q))\, =\, \int_{\C} z^p \overline{z}^q \mathcal{O}_n^{(1)}(z)\, d^2z 
+ \int_{\C} \int_{\C} z^p \overline{w}^q \mathcal{O}_n^{(2)}(z,w)\, d^2z\, d^2w\, .
\ee
Therefore, the mixed matrix moments are calculable from the overlap functions.
Moreover, the limiting values of the moments give some constraints for the limiting behavior of the overlap functions.
It is easy to see that $\mathcal{O}_n^{(1)}(e^{i\theta} z) = \mathcal{O}_n^{(1)}(z)$
and $\mathcal{O}_n^{(2)}(e^{i\theta}z,e^{i\theta}w) = \mathcal{O}_n^{(2)}(z,w)$,
consistent with the fact that $m_n((p);(q))$ equals $0$ unless $p=q$.

\section{Formulas for the overlap functions}

Chalker and Mehlig were able to relate the overlap functions to expectations of functions involving all
the eigenvalues. The eigenvalue distribution for the complex Ginibre ensemble is well-known.
In fact it is one of the simplest of the various Gaussian ensembles.
For example, as Chalker and Mehlig also point out in their paper,
\be
\label{eq:RD}
\mathcal{R}_N^{(1)}(z)\, =\, \frac{N}{\pi N!}\, e^{-N |z|^2} D_{N-1}(z)\, ,
\ee
where $D_{N-1}(z)$ equals the determinant of the $(N-1)$-dimensional square 
matrix $\mathcal{D}_N(z)$ where the matrix entries are best indexed for $j,k \in \{0,\dots,N-2\}$ as 
\be
[\mathcal{D}_N(z)]_{jk}\, =\, \left(\frac{N^{j+k+4}}{\pi^2 (j!) (k!)}\right)^{1/2} \int_{\C} \overline{\lambda}^j \lambda^k
|z-\lambda|^2 \exp(-N|\lambda|^2)\, 
d^2\lambda\, .
\ee
By rotational invariance of all the terms in the integrand other than $|z-\lambda|^2$, which is only
quadratic, it happens that $\mathcal{D}_N(z)$ is a tridiagonal matrix.
Hence Chalker and Mehlig point out that it is easy to derive a recursion relation for $D_{N-1}(z)$.
It is easier to define a new quantity $D_{N-1}(\sigma^{-2},z) = D_{N-1}(\sigma^{-1} N^{-1/2} z)$.
Then they show
\be
\label{eq:Dr}
D_{n+1}(\sigma^{-2},z)\, =\, (\sigma^{-2}|z|^2+n+1) D_n(\sigma^{-2},z) - \sigma^{-2} n |z|^2 D_{n-1}(\sigma^{-2},z)\, ,
\ee
and $D_0(\sigma^{-2},z)=1$, $D_1(\sigma^{-2},z) = 1+\sigma^{-2}|z|^2$. It turns out to be easy to solve this recurrence relation,
and Chalker and Mehlig give the formula
\begin{equation}
\label{eq:D}
D_{N-1}(\sigma^{-2},z)\, =\, (N-1)! \sum_{n=0}^{N-1} \frac{(\sigma^{-2}|z|^2)^n}{n!}\, ,
\end{equation}
which is the partial sum for the series for $(N-1)! \exp(\sigma^{-2}|z|^2)$. In order to obtain $D_{N-1}(z)$ one must take $\sigma^{-2}=N$.
So one sees that the dividing line is $|z|<1$ versus $|z|>1$, as to whether enough terms have been included in the partial sum
to get essentially $\exp(N|z|^2)$ or not. From this it follows that the measure $\mathcal{R}_N^{(1)}(z)\, d^2z$ converges 
weakly to $\pi^{-1} \mathbf{1}_{[0,1]}(|z|^2)\, d^2z$, as $N \to \infty$.
The reason for going into so much detail in this example is that the other examples are similar, but harder.
In fact, some of the formulas are so complicated that so far they have eluded any explicit, exact formula (at least as far as we have been able to find in the literature).

Another easy result which follows from these explicit formulas, but which does not appear in the paper of Chalker and Mehlig, is the scaling formula
near the unit circle. Let us record this for later reference.
\begin{lemma} \label{lem:circle}
 For any $u \in \R$,
\be
\mathcal{R}_N^{(1)}(1-N^{-1/2}u)\, \stackrel{N\to\infty}{\longrightarrow}\, \pi^{-1} \Phi(2u)\quad 
\text{ where }\quad
\Phi(x)\, =\, \frac{1}{\sqrt{2\pi}}\, \int_{-\infty}^{x} e^{-z^2/2}\, dz\, .
\ee
\end{lemma}
This is an important result which has been discovered before in \cite{Forrester,Kanzieper}.
Let us provide a quick proof for the reader's convenience.

\begin{proof}
Given the exact formula,
\be
\mathcal{R}_N^{(1)}(1-N^{-1/2}u)\, =\, \pi^{-1} \exp(-N+N^{1/2}u) \sum_{n=0}^{N-1} \frac{(N-N^{1/2}u)^n}{n!}\, ,
\ee
make the substitution $n=N-N^{1/2}x$ for $x \in \{N^{-1/2},2N^{-1/2},\dots,N^{1/2}\}$ and use Stirling's formula.
Then replace the sum by an appropriate integral in $x$ (of which it is a Riemann sum approximation with $\Delta x=N^{-1/2}$) by using the rigorous Euler-Maclaurin
summation formula.
\end{proof}
We may note that using the Euler-Maclaurin summation formula, one may obtain more terms as corrections of the leading-order term, just as one does
for the asymptotic series in Stirling's formula. Additionally, one may obtain formulas that are valid for more values of $u$:
one may obtain an asymptotic formula for $\mathcal{R}_N^{(1)}(z)-\pi^{-1}$ assuming that $|z|-1<CN^{-1/2}$ for some $C$, and another formula for $\mathcal{R}_N^{(1)}(z)$
assuming that $|z|-1>-CN^{-1/2}$ for some $C$: the difference in being whether one chooses to asymptotically evaluate the terms which are present in the partial sum for $\exp(N|z|^2)$
or whether one chooses to asymptotically evaluate the terms which are absent in that partial sum.

\subsection{More involved formulas:}
The formula for $\mathcal{O}_N^{(1)}$ is not much more complicated than the formula for $\mathcal{R}_N^{(1)}$, and Chalker and Mehlig gave the explicit answer.
It turns out that one may write $\mathcal{O}_N^{(1)}$ similarly to $\mathcal{R}_N^{(1)}$ as 
\be
\begin{gathered}
\mathcal{O}_N^{(1)}(z)\, 
=\, \frac{N}{\pi N!}\, \exp(-N|z|^2) G_{N-1}(z)\quad 
\text{ where }\quad
G_{N-1}(z)\, =\, \det[\mathcal{G}_{N-1}(z)]\, ,\\
\forall j,k \in \{0,\dots,N-2\}\, ,\quad
[\mathcal{G}_{N-1}(z)]_{jk}\, =\, \left(\frac{N^{j+k+4}}{\pi^2 (j!) (k!) }\right)^{1/2} \int_{\C} \overline{\lambda}^j \lambda^k
(N^{-1}+|z-\lambda|^2) \exp(-N|\lambda|^2)\, 
d^2\lambda\, .
\end{gathered}
\ee
The matrix $\mathcal{G}_{N-1}(z)$ is also tridiagonal for the same reason as $\mathcal{D}_{N-1}(z)$.
In particular, there is again a recursion relation for $G_{N-1}(z)$.
Defining $G_{N-1}(\sigma^{-2},z) = G_{N-1}(\sigma^{-1}N^{-1/2}z)$, one may see the recursion formula
\begin{equation}
\begin{split}
G_{n+1}(\sigma^{-2},z)\, 
&=\, [\mathcal{G}_n(\sigma^{-2},z)]_{nn} G_{n}(\sigma^{-2},z) - [\mathcal{G}_n(\sigma^{-2},z)]_{n,n-1} [\mathcal{G}_n(\sigma^{-2},z)]_{n-1,n} G_{n-1}(\sigma^{-2},z)\\
&=\, (\sigma^{-2}+n+2) G_{n}(\sigma^{-2},z)  - \sigma^{-2} n |z|^2 G_{n-1}(\sigma^{-2},z)\, ,
\end{split}
\end{equation}
with $G_0(\sigma^{-2},z) = 1$ and $G_1(\sigma^{-2},z) = 2 + \sigma^{-2}|z|^2$.
\begin{lemma}
\label{lem:G}
The exact solution to the recursion relation when $\sigma^{-2}=N$ is 
\be
\label{eq:lemG}
G_N(z)\, =\, (N-1)!\, \sum_{n=0}^{N-1} (N-n) \, \frac{(N|z|^2)^n}{n!}\, .
\ee
\end{lemma}
Using this formula, it is easy to see that $N^{-1} \mathcal{O}_N(z)\, d^2z$ converges weakly to $\pi^{-1} (1-|z|^2) \mathbf{1}_{[0,1]}(|z|^2)\, d^2z$, which is precisely the behavior
that Chalker and Mehlig found by other techniques. We will return to their approach, shortly.
For now, let us state the analogue of Lemma \ref{lem:circle}.
\begin{corollary}
For any $u \in \R$,
\be
\label{eq:Ocirc}
\mathcal{O}_N^{(1)}(1-N^{-1/2}u)\, \sim\, \frac{N^{1/2}}{\pi}\, \left[\frac{e^{-2u^2}}{\sqrt{2\pi}} - 2u\Phi(-2u)\right]\, ,\quad \text{ as $N \to \infty$.}
\ee
\end{corollary}
\begin{proof}
The proof is perfectly analogous to the proof of Lemma \ref{lem:circle}, except we start with Lemma \ref{lem:G} instead of equation (\ref{eq:D}).
\end{proof}
One also needs the two point function $\mathcal{O}_N^{(2)}$ in order to obtain any interesting moments.
The two-point function for the eigenvalues is easier to start with since its distribution is known exactly.
Using ideas related to the theory of orthogonal polynomials, one may see that $\mathcal{R}_N^{(2)}(z_1,z_2)$ is determinantal.
The canonical general reference for this is \cite{Mehta}.
One may write the formula as 
\be
\mathcal{R}_N^{(2)}(z_1,z_2)\, =\, \pi^{-2} e^{-N|z_1|^2} e^{-N|z_2|^2}
\det\left(K_N(z_j \overline{z}_k)\right)_{j,k=1}^{2}\quad \text{ for }\quad
K_N(z)\, =\, \sum_{n=0}^{N-1} \frac{(Nz)^n}{n!}\, .
\ee
From this one may determine the following asymptotics, proved in the same way as before.
\begin{lemma}
Define $\mathcal{C}_N^{(2)}(z_1,z_2) = \mathcal{R}_N^{(2)}(z_1,z_2) - \mathcal{R}_N^{(1)}(z_1) \mathcal{R}_N^{(1)}(z_2)$,
the corrected correlation function for the eigenvalues.
Then for any fixed $u_1,u_2 \in \C$
\be
\mathcal{C}_N^{(2)}(1-N^{-1/2}u_1,1-N^{-1/2}u_2)\, 
\sim\, \pi^{-2} e^{-|u_1-u_2|^2} |\Phi(-u_1-\overline{u}_2)|\, ,
\ee
where the definition of $\Phi$ is extended to the complex plane as $\Phi(-u)\, =\, (2\pi)^{-1/2} e^{-u^2/2} \int_{0}^{\infty} e^{-x^2/2} e^{-ux}\, dx$.
\end{lemma}
We  have stated a somewhat precise limit for $\mathcal{R}_N^{(2)}$. But we do not know how to get a precise limit for $\mathcal{O}_N^{(2)}$.
Let us state one of Chalker and Mehlig's main results as a conjecture. In other words, they give a good argument for the calculation of $\mathcal{O}_N^{(2)}$
which is highly plausible on the basis of mathematical reasoning. But to the best of our knowledge their result has not yet been fully rigorously proved.
\begin{conj}[Chalker and Mehlig]
(i) For any two points $z_1,z_2$ such that $|z_1|<1$, $|z_2|<1$ and $|z_1-z_2|>0$,
\be
\mathcal{O}_N^{(2)}(z_1,z_2)\, \stackrel{N\to\infty}{\longrightarrow}\,
-\frac{1}{\pi^2}\cdot \frac{1-z_1\overline{z}_2}{|z_1-z_2|^4}\, .
\ee
(ii)
For any $\omega \in \C$ and $z$ such that $|z|<1$,
\be
\label{eq:CM2}
N^{-2} \mathcal{O}_N^{(2)}\Big(z+\frac{1}{2}N^{-1/2}\omega,z-\frac{1}{2}N^{-1/2}\omega\Big)\,
\sim\, -\pi^{-2} (1-|z|^2)\, \frac{1-(1+|\omega|^2)e^{-|\omega|^2}}{|\omega|^4}\, ,\qquad \text{ as $N \to \infty$.}
\ee
\end{conj}
Importantly, there is no asymptotic formula for $z_1$ and $z_2$ near the boundary of the circle.
For all the other cases, this regime gives lower-order corrections, beyond the leading order.

Chalker and Mehlig's approach is beautiful and compelling.
They calculated an explicit formula for $\mathcal{O}_N^{(2)}(0,z)$.
Note, for instance, that $\mathcal{R}_N^{(2)}(0,z) = \pi^{-1} (\mathcal{R}_N^{(1)}(z)-\pi^{-1}e^{-|z|^2})$.
So the formula simplifies when one of the arguments is $0$.
A similar fact holds for $\mathcal{O}_N^{(2)}(z_1,z_2)$, even though it seems that it is not determinantal
like $\mathcal{R}_N^{(2)}(z_1,z_2)$ is.
Then Chalker and Mehlig considered a universality-type argument to see how the functional form should behave under transformations
of the point $0$ to other places on the circle.
Their argument is also a universal argument, applying to more ensembles than just the complex Ginibre ensemble.
But we will continue to consider just the complex Ginibre ensemble, here.

The second part of their argument is the key to their formula. 
The function $\mathcal{O}_N^{(2)}(z_1,z_2)$ may be expressed as the expectation
of a non-local function of all the eigenvalues of $A_n$.
Chalker and Mehlig observe that the function depends mainly on the eigenvalues in a core small
area around $z_1$ and $z_2$.
For this core, the distribution of the eigenvalues should be universal, not depending on the proximity
of $z_1$ and $z_2$ to the boundary of the disk, as long as they are not near the boundary.
Then outside the core there is a self-averaging contribution of all the other eigenvalues, which
may be reduced to a Riemann integral approximation, and calculated.
That part does depend on the geometry of the point configuration in the disk, but it is easily calculated.
Putting these two parts together with their formula for $\mathcal{O}_N^{(2)}(0,z)$,
they were able to arrive at (\ref{eq:CM2}).

The reader is advised most strongly to consult their beautiful paper.

Now we want to explain briefly the first part of their argument since it is a basis for a different proposal
we have for how to prove their conjecture.
Chalker and Mehlig point out that $\mathcal{O}_N^{(2)}(z_1,z_2)$ may be calculated as the determinant of a 5-diagonal matrix.
In fact, it is easier to start with $\mathcal{R}_N^{(2)}(z_1,z_2)$:
\be
\mathcal{R}_N^{(2)}(z_1,z_2)\, 
=\, \frac{N^3}{\pi^2 N!}\, |z_1-z_2|^2 e^{-N|z_1|^2} e^{-N|z_2|^2} F_{N-2}(z_1,z_2)\, ,
\ee
where $F_{N-2}(z_1,z_2)$ equals the determinant of the $(N-2)$-dimensional square matrix $\mathcal{F}_{N-2}(z_1,z_2)$, where
\be
[\mathcal{F}_{N-2}(z_1,z_2)]_{jk}\, =\, \left(\frac{N^{j+k+6}}{\pi^2 (j+1)! (k+1)! }\right)^{1/2} \int_{\C} \overline{\lambda}^j \lambda^k |z_1-\lambda|^2 |z_2-\lambda|^2 
\exp(-N|\lambda|^2)\, d^2\lambda\, ,
\ee
for $j,k=0,\dots,N-3$. Then the formula for $\mathcal{O}_N^{(2)}(z_1,z_2)$ is 
\be
\mathcal{O}_N^{(2)}(z_1,z_2)\, 
=\, -\frac{N^2}{\pi^2 N!}\, e^{-N|z_1|^2} e^{-N|z_2|^2} H_{N-2}(z_1,z_2)\, ,
\ee
where $H_{N-2}(z_1,z_2)$ equals the determinant of the $(N-2)$-dimensional square matrix $\mathcal{H}_{N-2}(z_1,z_2)$, where
\be
[\mathcal{H}_{N-2}(z_1,z_2)]_{jk}\, =\, \left(\frac{N^{j+k+6}}{\pi^2 (j+1)! (k+1)! }\right)^{1/2} \int_{\C} \overline{\lambda}^j \lambda^k \Big[|z_1-\lambda|^2 |z_2-\lambda|^2 
+ N^{-1} \left(\overline{z}_1 - \overline{\lambda}\right)(z_2-\lambda)\Big]
\exp(-N|\lambda|^2)\, d^2\lambda\, ,
\ee
for $j,k=0,\dots,N-3$.
These are naturally 5-diagonal because of rotation invariance.
But if $z_1=0$ or $z_2=0$ then they become tri-diagonal, again.
Hence they are more easily calculable in that case.
That is why $\mathcal{O}_N^{(2)}(0,z)$ is calculable.

In Section 6 we are going to propose another method to proceed.
This is to write down the recursion relation for the 5-diagonal matrix, which is harder than for a tridiagonal
matrix.
Then, even if the formula is not exactly solvable, we argue that it should be asymptotically
solvable.
We give more details in Section 6, in particular carrying out the asymptotic approach for the easier
problem of calculating $\mathcal{R}_N^{(2)}(z)$ (which we may check against the exact solution).

\section{Moments and constraints on the overlap functions}
An ideal situation would be to find an explict sum-formula for $\mathcal{O}_N^{(2)}(z_1,z_2)$, just as Lemma \ref{lem:G}
provides for $\mathcal{O}_N^{(1)}(z)$. 
But so far, this has not been discovered.
In the Section 6 we will suggest a rigorous approach which may work to give the asymptotics, even when no explicit formula is known.
But for now, let us state the constraints imposed by the moment formula from before.

Recall from (\ref{eq:mnpq})
for any nonnegative integers $p$ and $q$,
$$
m_N((p);(q))\, =\, \int_{\C} z^p \overline{z}^q \mathcal{O}_N^{(1)}(z)\, d^2z 
+ \int_{\C} \int_{\C} z^p \overline{w}^q \mathcal{O}_N^{(2)}(z,w)\, d^2z\, d^2w\, .
$$
But moreover, from the discussion at the end of Section \ref{subsec:Argument},
$m_N((p);(q))$ equals $0$ unless $p=q$, and as noted at the end of Section \ref{sec:Expected}, this is already reflected in the rotational invariance properties of
$\mathcal{O}_N^{(1)}(z)$ and $\mathcal{O}_N^{(2)}(z_1,z_2)$.
Therefore, specializing, we see that
\be
\label{eq:MoTot}
\int_{\C} |z|^{2p} \mathcal{O}_N^{(1)}(z)\, d^2z 
+ \int_{\C} \int_{\C} z_1^p \overline{z}_2^p \mathcal{O}_N^{(2)}(z_1,z_2)\, d^2z_1\, d^2z_2\, =\, 1\, ,
\ee
for each nonnegative integer $p$.
This is the constraint formula.
Let us now analyze this formula, starting with the leading order terms, and going down in order.

\subsection{Cancelling divergences at leading order}

For any fixed $z$ with $|z|<1$, we have
\be
\mathcal{O}_N^{(1)}(z)\, \sim\, N \pi^{-1} (1-|z|^2)\, ,
\ee
and the corrections are actually exponentially small in $N$ (since they arise as the deep part of the right tail of the series for the exponential).
Therefore, integrating, we obtain the leading-order part of the contribution from $\mathcal{O}_N^{(1)}(z)$ from the formula above
\be
\int_{\C} |z|^{2p} \mathcal{O}_N^{(1)}(z)\, d^2z\, \sim\, N \pi^{-1} \int_{\C} |z|^{2p}(1-|z|^2) \mathbf{1}_{[0,1)}(|z|^2)\, d^2z\, .
\ee
The corrections to this formula are not exponentially small, incidentally.
This is because the formula for $\mathcal{O}_N^{(1)}(z)$ is not exponentially close to the exact formula for all $z$ in the complex plane.
For a fixed $|z|>1$ it is easy to see that $\mathcal{O}_N^{(1)}(z)$ is exponentially small (hence exponentially close to the approximating function of $0$ there).
That is because one only has the series for the exponential up to a small number of terms, deep in the left tail.
But near the circle, there are algebraic corrections, not exponential ones.

Nevertheless, let us note that, by making a polar decomposition, $z=r e^{i\theta}$, we obtain
\be
N \pi^{-1} \int_{\C} |z|^{2p}(1-|z|^2) \mathbf{1}_{[0,1)}(|z|^2)\, d^2z\, 
=\, N \int_0^1 t^p (1-t)\, dt\, =\, \frac{N}{(p+1)(p+2)}\, .
\ee
Let us see how this cancels with the leading-order part of the $\mathcal{O}_N^{(2)}$ integral.

We will use Chalker and Mehlig's formula here for the leading-order part, even though we do not yet know
the corrections for the lower-order part near the circle.
Then we get
\begin{multline}
\label{eq:leadingO2}
\int_{\C} \int_{\C} z_1^p \overline{z}_2^p \mathcal{O}_N^{(2)}(z_1,z_2)\, d^2z_1\, d^2z_2\\
\sim\, - \int_{\C} \int_{\C} \left(z+\frac{\omega}{2N^{1/2}}\right)^p \left(\overline{z}-\frac{\overline{\omega}}{2N^{1/2}}\right)^p N^2 \pi^{-2} (1-|z|^2) \mathbf{1}_{[0,1)}(|z|^2)\, \frac{1-(1+|\omega|^2)e^{-|\omega|^2}}{|\omega|^4}\, N^{-1} d^2\omega\, d^2z\, ,
\end{multline}
where the $N^{-1}$ associated to the volume-element $d^2\omega$-times-$d^2z$ is to account for the Jacobian of the transformation from $(z_1,z_2)$ to $(z,\omega)$.
Now we will begin to separate this formula into even another decomposition into leading terms, and sub-leading terms.
This is because, in the formulas $z_1^p = (z+\frac{1}{2}N^{-1/2}\omega)^p$ and $\overline{z}_2^p = (\overline{z}-\frac{1}{2}N^{-1/2}\overline{\omega})^p$,
clearly the leading order arises by ignoring the contributions of $\omega$ which each are accompanied by negative powers of $N$.
So we really obtain, as what we might call the ``leading order, leading order'' term:
\be
\int_{\C} \int_{\C} z_1^p \overline{z}_2^p \mathcal{O}_N^{(2)}(z_1,z_2)\, d^2z_1\, d^2z_2\,
\sim\, - N \pi^{-2} \int_{\C} \int_{\C} |z|^{2p}(1-|z|^2) \mathbf{1}_{[0,1)}(|z|^2)\, \frac{1-(1+|\omega|^2)e^{-|\omega|^2}}{|\omega|^4}\, N^{-1} d^2\omega\, d^2z\, .
\ee
Then it is easy to see that this splits. The integral over $\omega$ is 
\be
\int_{\C} \frac{1-(1+|\omega|^2)e^{-|\omega|^2}}{|\omega|^4}\, d^2\omega\,
=\, \pi \int_0^{\infty} \frac{1-(1+t)e^{-t}}{t^2}\, dt\,
=\, \pi \int_0^{\infty} \frac{1}{t^2} \left(\int_0^t s e^{-s}\, ds\right)\, dt\, .
\ee
Integrating-by-parts, it is easy to see that this gives $\pi$.
Therefore, we end up with the exact negative of the leading order contribution by $\mathcal{O}_N^{(1)}$:
\be
\int_{\C} \int_{\C} z_1^p \overline{z}_2^p \mathcal{O}_N^{(2)}(z_1,z_2)\, d^2z_1\, d^2z_2\,
\sim\, - N \pi^{-1} \int_{\C} |z|^{2p}(1-|z|^2) \mathbf{1}_{[0,1)}(|z|^2)\,  d^2z\, =\, -\frac{N}{(p+1)(p+2)}\, .
\ee
The fact that these two terms cancel is good, because each diverges, separately; whereas, according to the formula, the exact answer is supposed to be $1$.

\subsection{The sub-leading contribution from $\mathcal{O}_N^{(1)}$}
\label{subsec:sLO1}
For the first integral, we are fortunate that the exact correction is known near the circle.
We will not attempt to keep track of the exponentially-small corrections which are present away from the circle.
But near the circle, the exact corrections are relevant because they are not exponentially small.

Using Corollary 4.3, we know that
\be
\int_{\C} |z|^{2p} \mathcal{O}_N^{(1)}(z)\, d^2z
- \frac{N}{(p+1)(p+2)}\, =\, \frac{N^{1/2}}{\pi}\, \int_{\C} \left[\frac{e^{-2u^2}}{\sqrt{2\pi}} - 2 u \Phi(-2u)\right]\Bigg|_{u = N^{1/2} (1-|z|)} |z|^{2p}\, d^2z
+ o(1)\, ,
\ee
where the small term $o(1)$ means that the remainder converges to $0$ as $N \to \infty$.
This remainder includes exponentially small corrections to $\mathcal{O}_N^{(1)}(z)$ away from the circle, as well as the systematic correction terms to the leading-order
behavior near the circle that arise from the Euler-Maclaurin series.
The reason that these correction terms to the Euler-Maclaurin summation formula are $o(1)$ will arise momentarily: even the leading order term is only order-1, constant.

Making the polar decomposition of $z$ and then rewriting $r = 1-N^{-1/2}u$ so that $dr=N^{-1/2}\, du$ (and reversing orientation of the integral), we have
\be
\begin{split}
\int_{\C} |z|^{2p} \mathcal{O}_N^{(1)}(z)\, d^2z
- \frac{N}{(p+1)(p+2)}\, 
&=\, 2 \int_{-\infty}^{N^{1/2}} \left[\frac{e^{-2u^2}}{\sqrt{2\pi}} - 2 u \Phi(-2u)\right](1-N^{-1/2}u)^{2p+1}\, du+ o(1)\\
&=\, 2 \int_{-\infty}^{\infty} \left[\frac{e^{-2u^2}}{\sqrt{2\pi}} - 2 u \Phi(-2u)\right]\, du+ o(1)\, .
\end{split}
\ee
In particular, this correction is independent of $p$, modulo vanishingly small remainder terms which are accumulated in the $o(1)$.
Rewriting $u=x/2$ and integrating by parts gives a constant which is equal to $3/2$.

We will not be able to make it to the order-1, constant terms in the $N\to \infty$ asymptotics series (in decreasing powers of $N$).
The reason is that for $\mathcal{O}_N^{(2)}$, we do not have sufficiently precise asymptotics to get to that level.
Instead, what we will do next is to consider what constraints the formula for the moments imposes on $\mathcal{O}_N^{(2)}$.

\subsection{Sub-leading divergences in the $\mathcal{O}_N^{(2)}$ term}

We have now accounted for all the non-vanishing contributions from the $\mathcal{O}_N^{(1)}$ term.
The leading-order divergence cancels with the leading-order divergence of the $\mathcal{O}_N^{(2)}$ term.
The sub-leading order part of the $\mathcal{O}_N^{(1)}$ contribution to the moment is already order-1, constant, and it is independent of $p$.
It equals $3/2$.
Note that the moment itself is also independent of $p$, it is $1$.

Since we do not know the actual formula for $\mathcal{O}_N^{(2)}$, our plan for this section is to consider the proposed formula
for $\mathcal{O}_N^{(2)}$ in the bulk.
That still leads to one other divergent contribution, diverging logarithmically in $N$.
What this must mean is that in the formula for $\mathcal{O}_N^{(2)}(z_1,z_2)$ for $z_1$ and $z_2$
close, and both near the circle, there must be an edge correction,
which leads to a counter-balancing divergence.
This is what we explain in some more detail, now.
This subsection is detailed and technical.

We consider the proposed formula for $\mathcal{O}_N^{(2)}$ that Chalker and Mehlig derived.
This is the correct formula in the bulk, following the argument of their paper, although there is a lower-order correction near the circle.
We will not include the correction on the circle. Instead our calculations will show constraints that must be satisfied for this correction formula.
We use $z_1 = z+\frac{1}{2}N^{-1/2}\omega$ and $z_2 = z-\frac{1}{2}N^{-1/2}\omega$ so that
\be
z_1 \overline{z}_2\,
=\, |z|^2 + \frac{i \mathrm{Im}[\omega \overline{z}]}{N^{1/2}} - \frac{|\omega|^2}{4N}\, .
\ee
Therefore, using the bulk formula we would have
\be
\begin{split}
\int_{\C} \int_{\C} z_1^p \overline{z}_2^p \mathcal{O}_N^{(2)}(z_1,z_2)\, d^2z_1\, d^2z_2\,
&\approx\, -\frac{N^2}{\pi^2} \int_{\C} \int_{\C} \left[|z|^2 + \frac{i \mathrm{Im}[\omega \overline{z}]}{N^{1/2}} - \frac{|\omega|^2}{4N}\right]^p
\left(1 -  \left[|z|^2 + \frac{i \mathrm{Im}[\omega \overline{z}]}{N^{1/2}} - \frac{|\omega|^2}{4N}\right]\right)\\
&\hspace{2cm}
\cdot \frac{1-(1+|\omega|^2)e^{-|\omega|^2}}{|\omega|^4}\,
\mathbf{1}_{[0,1]}\Big(\Big|z\pm\frac{1}{2}N^{-1/2}\omega\Big|^2\Big)\, 
N^{-1} d^2\omega\, d^2z\, ,
\end{split}
\ee
where we use the approximation symbol $\approx$ to remind ourselves that this is only one part of the eventual formula.
Simplifying this, and writing $z = r e^{i\theta}$ and $\omega = \rho e^{i t}$, we have
\be
\begin{split}
\int_{\C} \int_{\C} z_1^p \overline{z}_2^p \mathcal{O}_N^{(2)}(z_1,z_2)\, d^2z_1\, d^2z_2\,
&\approx\, -\frac{N}{\pi^2} \int_{\C} \int_{\C} \left[r^2 + \frac{i r\rho \sin(\theta-t)}{N^{1/2}} - \frac{\rho^2}{4N}\right]^p
\left(1 -  \left[r^2 + \frac{i r\rho \sin(\theta-t)}{N^{1/2}} - \frac{\rho^2}{4N}\right]\right)\\
&\hspace{1.75cm}
\cdot \frac{1-(1+\rho^2)e^{-\rho^2}}{\rho^4}\,
\mathbf{1}_{[0,4N]}\Big(\Big|\rho \pm 2N^{1/2}re^{i(t-\theta)}\Big|^2\Big)\, 
r\rho\, dr\, d\rho\, d\theta\, dt\, .
\end{split}
\ee
Let us denote $\phi = t-\theta$. Integrating over the extra angular variable, and simplifying the power of $\rho$ in the second line, and simplifying the indicator in the second line, we obtain
\be
\begin{split}
\int_{\C} \int_{\C} z_1^p \overline{z}_2^p \mathcal{O}_N^{(2)}(z_1,z_2)\, d^2z_1\, d^2z_2\,
&\approx\, -\frac{2N}{\pi} \int\limits_{r \in [0,1]} \int\limits_{\rho>0} \int\limits_{\phi \in [0,2\pi)} \left[r^2 - \frac{i r\rho \sin(\phi)}{N^{1/2}} - \frac{\rho^2}{4N}\right]^p
\left(1 -  \left[r^2 - \frac{i r\rho \sin(\phi)}{N^{1/2}} - \frac{\rho^2}{4N}\right]\right)\\
&\hspace{2.5cm}
\cdot \frac{1-(1+\rho^2)e^{-\rho^2}}{\rho^3}\,
\mathbf{1}_{[0,N^{1/2} R(r,\phi)]}(\rho)\, 
r dr\, d\rho\, d\phi\, ,
\end{split}
\ee
 for
\be
R(r,\phi)\, =\, 2 \left(\sqrt{1-r^2\sin^2(\phi)}-r|\cos(\phi)|\right)\, ,
\ee
arising from the condition
$|\rho \pm 2 N^{1/2} r e^{i\phi}| \leq 2N^{1/2} \Leftrightarrow \rho \leq R(r,\phi) N^{1/2}$

Let us rewrite this once again, this time isolating different functional terms that we wish to consider in more detail:
\be
\int_{\C} \int_{\C} z_1^p \overline{z}_2^p \mathcal{O}_N^{(2)}(z_1,z_2)\, d^2z_1\, d^2z_2\,
\approx\, -\frac{2N}{\pi} \int\limits_{r \in [0,1]} \int\limits_{\phi \in [0,2\pi)} \left(\int_0^{N^{1/2}R(r,\phi)} F_p(r,\phi,\rho) W(\rho)\, d\rho\right)\, r\, dr\, d\phi\, ,
\ee
where
\be
\label{eq:Fp}
F_p(r,\phi,\rho)\, =\, \left[r^2 - \frac{i r\rho \sin(\phi)}{N^{1/2}} - \frac{\rho^2}{4N}\right]^p
\left(1 -  \left[r^2 - \frac{i r\rho \sin(\phi)}{N^{1/2}} - \frac{\rho^2}{4N}\right]\right)\, ,
\ee
and
\be
W(\rho)\, =\,  \frac{1-(1+\rho^2)e^{-\rho^2}}{\rho^3}\, .
\ee
Now we note that we can expand
\be
\label{eq:expand}
F_p(r,\phi,\rho)\, =\, \sum_{k=0}^{2p} f_p^{(k)}(r,\phi)\, \frac{\rho^k}{N^{k/2}}\, .
\ee
Odd powers of $k$ have $f_p^{(k)}(r,\phi)$ which is an odd function of $\sin(\phi)$. Since the rest of the integral will contribute even factors, this means all odd powers will integrate to zero.
So we only keep track of even powers.
We have already taken account of $f_p^{(0)}(r,\phi)$ which is just $f_p^{(0)}(r) = r^{2p} (1-r^2)$.
This was what gave us the leading order divergence we considered in Subsection 5.1, in equation (\ref{eq:leadingO2}).

Moreover, starting from the even power $k=4$, we have
\be
-\frac{2N}{\pi} \int\limits_{r \in [0,1]} \int\limits_{\phi \in [0,2\pi)} \left(\int_0^{N^{1/2}R(r,\phi)} f_p^{(k)}(r,\phi)\, \frac{\rho^k}{N^{k/2}} W(\rho)\, d\rho\right)\, r\, dr\, d\phi\, 
=\, O(1)\, .
\ee
The reason is that $N \cdot N^{-k/2} = N^{-(k-2)/2}$ which is vanishing. This means that the lower limit of integration is actually contributing a negligible correction, asymptotically for large $N$.
But near the upper limit of integration we may expand $W(\rho) \sim \rho^{-3}$.
Therefore we obtain near the upper limit, for $k=4,6,\dots$,
\begin{multline}
\label{eq:largek}
-\frac{2}{\pi N^{(k-2)/2}} \int\limits_{r \in [0,1]} \int\limits_{\phi \in [0,2\pi)} f_p^{(k)}(r,\phi)\,\left(\int_0^{N^{1/2}R(r,\phi)}  \rho^{k-3}\, d\rho\right)\, r\, dr\, d\phi\\
=\, -\frac{2}{\pi} \int\limits_{r \in [0,1]} \int\limits_{\phi \in [0,2\pi)} f_p^{(k)}(r,\phi)\, \frac{[R(r,\phi)]^{(k-2)/2}}{k-2}\, r\, dr\, d\phi
+o(1)\, =\, O(1)\, .
\end{multline}

This only leaves the term with $k=2$ which might diverge.
Indeed, for this, we just have
$$
-\frac{2}{\pi}\,
\int\limits_{r \in [0,1]} \int\limits_{\phi \in [0,2\pi)} f_p^{(2)}(r,\phi)\,\left(\int_0^{N^{1/2}R(r,\phi)}  \frac{1-(1+\rho^2)e^{-\rho^2}}{\rho}\, d\rho\right)\, r\, dr\, d\phi\, .
$$
The only divergent part of this arises near the upper limit for the $\rho$ integral which gives $\ln(N^{1/2} R(r,\phi)) = \frac{1}{2} \ln(N) + \ln(R(r,\phi))$.
So the logarithmic divergence is 
\begin{multline}
-\frac{2}{\pi}\,
\int\limits_{r \in [0,1]} \int\limits_{\phi \in [0,2\pi)} f_p^{(2)}(r,\phi)\,\left(\int_0^{N^{1/2}R(r,\phi)}  \frac{1-(1+\rho^2)e^{-\rho^2}}{\rho}\, d\rho\right)\, r\, dr\, d\phi\\
=\,
-\frac{\ln(N)}{\pi} \int\limits_{r \in [0,1]} \int\limits_{\phi \in [0,2\pi)} r f_p^{(2)}(r,\phi)\, dr\, d\phi
+ O(1)\, .
\end{multline}
It is easy to see that
\be
f_p^{(2)}(r,\phi)\, =\,
\frac{1}{4}\, r^{2p} - \frac{p}{4}\, r^{2p-2} (1-r^2) - \frac{p(p-1)}{2}\, r^{2p-2}(1-r^2) \sin^2(\phi) + p r^{2p}\sin^2(\phi)\, .
\ee
Therefore, we have
\be
\begin{split}
\int_{\phi \in [0,2\pi]} f_p^{(2)}(r,\phi)\, d\phi\, 
&=\, 2\pi
\left(\frac{1}{4}\, r^{2p} - \frac{p}{4}\, r^{2p-2} (1-r^2) - \frac{p(p-1)}{4}\, r^{2p-2}(1-r^2) + \frac{p}{2}\, r^{2p}\right)\\
&=\, \frac{\pi}{2} \left[(p+1)^2 r^{2p} - p^2 r^{2p-2}\right]\, .
\end{split}
\ee
Therefore, the sub-leading order divergence is now
\be
\label{eq:sL}
\int_{\C} \int_{\C} z_1^p \overline{z}_2^p \mathcal{O}_N^{(2)}(z_1,z_2)\, d^2z_1\, d^2z_2
+ \frac{N}{(p+1)(p+2)}\,
=\, -\frac{1}{4}\, \ln(N) + O(1)\, .
\ee
We may consider this particular form.
It is independent of $p$.
Near the circle, and for $z_1$ near $z_2$, the form of $z_1^p \overline{z}_2^p$, to leading order is just $|z|^{2p}$ which is just $1$, because $z$ is near the circle.
This is the same explanation for the reason that the order-1, constant term coming from $\mathcal{O}_N^{(1)}$ term is independent of $p$.
We also know that the moment must be independent of $p$.

One could also try to calculate the order-1 contributions at this point, coming just from
the bulk formula for $\mathcal{O}_N^{(2)}(z_1,z_2)$.
One could then check whether these combine to a constant independent of $p$.
That would be yet another strong check that Chalker and Mehlig's formula for $\mathcal{O}_N^{(2)}(z_1,z_2)$  is tru to very high accuracy in the bulk, and only needs an edge correction
near the circle.

It would be best to have a sufficiently explict formula for $\mathcal{O}_N^{(2)}(z_1,z_2)$ to allow one to see the correction near the circle.
Then we could have an answer to settle this.
Next we propose a method which we believe could potentially provide this.

\section{Proposal to rigorously approach Chalker and Mehlig's result}
There are various ways to try to prove Chalker and Mehlig's formula for the bulk behavior of $\mathcal{O}_N^{(2)}$.
One way is to try to fill in the details to make Chalker and Mehlig's argument rigorous.
Their idea is to express $\mathcal{O}_N^{(2)}$ in terms of the expectation of a function of the eigenvalues, and then use the known eigenvalue marginal
for the complex Ginibre ensemble.
Let us point out that the approach of Chalker and Mehlig has been seriously studied and followed up by the group of Nowak, and others, in \cite{Janick}.

Here we want to propose a different method.
The formula for $\mathcal{O}_N^{(2)}(z_1,z_2)$ is the determinant of a 5-diagonal matrix.
One may express such a determinant through a recursion relation, although the recursion relation is significantly more complicated than in the tridiagonal case.
It is higher order, and it is a vector valued recursion relation for a vector with dimension greater than $1$.
We will not explicate this, here.
It is well-known, it just follows from Cramer's rule, and it is used in numerical codes in wide use.

Instead, what we want to advocate here is solving recursion relations, at least asymptotically for large $N$, using adiabatic theory.
We have not tried this yet for $\mathcal{O}_N^{(2)}(z_1,z_2)$.
There may be formidable difficulties which obstruct this approach.
But let us demonstrate the idea for an easier problem: re-deriving the formula for $\mathcal{R}_N^{(1)}(z)$.
This leads to an easier problem.
The key trick for this particular problem is to realize that $\mathcal{R}_N^{(1)}(z)$, at least for the leading-order asymptotic formula,
is constant in $z$ for $|z|<1$.

\subsection{The recurrence relation for $\mathcal{R}_N^{(1)}(z)$ using matrices}

We are treating the case of $\mathcal{R}_N^{(1)}(z)$ as a simpler toy model,
in lieu of treating the real problem of interest which is $\mathcal{O}_N^{(2)}(z_1,z_2)$.
We hope to be able to handle $\mathcal{O}_N^{(2)}(z_1,z_2)$ later, in another paper.

Recall from (\ref{eq:RD}) that $\mathcal{R}_N^{(1)}(z) = \pi^{-1} [(N-1)!]^{-1} \exp(-N |z|^2) D_{N-1}(z)$, which means from Stirling's formula that
\be
\mathcal{R}_N^{(1)}(z)\, \sim\, \frac{1}{\pi} \cdot \frac{1}{\sqrt{2 \pi N}}\, e^{-(N-1)\ln(N)+N(1-|z|^2)} D_{N-1}(z)\, .
\ee
Moreover, recall that there is a recursion relation in (\ref{eq:Dr}). Namely, defining $D_{N-1}(\sigma^{-2},z) = D_{N-1}(\sigma^{-1}N^{-1/2}z)$,
it happens that
$$
D_{n+1}(\sigma^{-2},z)\, =\, (\sigma^{-2}|z|^2+n+1) D_n(\sigma^{-2},z)-\sigma^{-2}n|z|^2D_{n-1}(\sigma^{-2},z)\, .
$$
Let us fix $\sigma^2=N^{-1}$ as Chalker and Mehlig do.
So
\be
D_{n+1}(N,z)\, =\, (N|z|^2+n+1) D_n(N,z)-Nn|z|^2D_{n-1}(N,z)\, .
\ee
Also, since the answer only depends on the magnitude of $z$, let us write $r=|z|$ so
\be
D_{n+1}(N,r)\, =\, (Nr^2+n+1) D_n(N,r)-Nnr^2D_{n-1}(N,r)\, .
\ee
And we want to calculate $\mathcal{R}_N^{(1)}(r)$ which is asymptotically given by
\be
\mathcal{R}_N^{(1)}(r)\, \sim\, \frac{1}{\pi} \cdot \frac{1}{\sqrt{2 \pi N}}\, e^{-(N-1)\ln(N)+N(1-r^2)} D_{N-1}(N,r)\, .
\ee
Now since we have a second-order recursion relation, let us define a two-dimensional vector $v_n = [D_{n-1}(N,r),D_n(N,r)]^*$.
(All our vectors and matrices will be real but we use the adjoint instead of the transpose because we want to keep the symbol
$T$ for other purposes.)
Then the recursion relation says that
\be
v_{n+1}\, =\, A_n v_n\, ,\qquad
A_n\, =\, 
\begin{bmatrix} 0 & 1 \\
-Nnr^2 & N r^2+n+1
\end{bmatrix}\, ,
\ee
and we want $D_{N-1}(N,r) = e_2^* v_{N-1}$, where $\{e_1,e_2\}$ is the standard basis for $\R^2$.
In order to have a simpler formula, we note that we can write $v_1 = A_0 e_2$, for $A_0$ defined as above. 
So in seeking $\mathcal{R}_N^{(1)}(r)$, we really have
\be
\mathcal{R}_N^{(1)}(r)\, \sim\, \frac{1}{\pi} \cdot \frac{1}{\sqrt{2 \pi N}} \, e^{-(N-1)\ln(N)+N(1-r^2)} e_2^* A_{N-2} \cdots A_1 A_0 e_2\, .
\ee
The idea is to try to express this using the spectral decomposition of the matrices $A_n$,
where we use the fact that the matrices $A_n$ are {\em varying slowly} in $n$, as much as possible.
This is why we call this the {\em adiabatic approach.}

\subsection{Spectral formulas and summary of main contribution}

We may summarize the spectral information as 
\be
\begin{gathered}
\lambda_n^{\pm}\, =\, \frac{N}{2} \left(r^2+\frac{n+1}{N} \pm \sqrt{\left(\frac{n+1}{N}-r^2\right)^2+4r^2N^{-1}}\right)\, ,\quad
V_n^{\pm}\, =\, \begin{bmatrix} 1 \\ \lambda_n^{\pm} \end{bmatrix}\, ,\quad
W_n^{\pm}\, =\, \pm\frac{1}{\lambda_n^+ - \lambda_n^-} \begin{bmatrix} - \lambda_n^{\pm} \\ 1 \end{bmatrix}\, ;\\
A_n V_n^{\pm}\, =\, \lambda_n^{\pm} V_n^{\pm}\, ,\qquad
(W_n^{\pm})^* A_n\, =\, \lambda_n^{\pm} (W_n^{\pm})^*\, ,\qquad
(W_n^{\sigma})^* V_n^{\tau}\, =\, \delta_{\sigma,\tau}\, ,\ \text{ for $\sigma,\tau \in \{+1,-1\}$.}
\end{gathered}
\ee
In particular, $A_n = \lambda_n^+ V_n^+ (W_n^+)^* + \lambda_n^- V_n^- (W_n^-)^*$.
Therefore, we can rewrite the conclusion of the recursion relation as 
\be
\begin{split}
\mathcal{R}_N^{(1)}(r)\, 
&\sim\, \frac{1}{\pi} \cdot \frac{1}{\sqrt{2 \pi N}} \, e^{-(N-1)\ln(N)+N(1-r^2)} e_2^* A_{N-2} \cdots A_1 A_0  e_2\\
&=\, \frac{1}{\pi} \cdot \frac{1}{\sqrt{2 \pi N}} \, e^{-(N-1)\ln(N)+N(1-r^2)}
\sum_{\sigma \in \{+1,-1\}^{N-1}} \left(\prod_{n=0}^{N-2} \lambda_n^{\sigma(n)}\right)
([W_0^{\sigma(0)}]^* e_2) (e_2^* V_{N-2}^{\sigma(N-2)}) \\
&\hspace{9cm}
\cdot \left(\prod_{n=0}^{N-1} [W_{n+1}^{\sigma(n+1)}]^* V_n^{\sigma(n)}\right)\, .
\end{split}
\ee
Anticipating that the main contribution to this sum will be $\sigma(0)=\dots=\sigma(N-1)=+1$, we may rewrite this as 
\be
\mathcal{R}_N^{(1)}(r)\, 
\sim\, \frac{1}{\pi} \cdot \frac{1}{\sqrt{2 \pi N}} \, e^{-(N-1)\ln(N)+N(1-r^2)} \mathcal{M}_N(r) \mathcal{P}_N(r)\, ,
\ee
where $\mathcal{M}_N(r)$ is the ``main term''
\be
\mathcal{M}_N(r)\,
=\,
\left(\prod_{n=0}^{N-2} \lambda_n^{+}\right)
([W_0^{+}]^* e_2) (e_2^* V_{N-2}^{+}) 
\cdot \left(\prod_{n=0}^{N-1} [W_{n+1}^{+}]^* V_n^{+}\right)\, ,
\ee
and $\mathcal{P}_N(r)$ will be a series of perturbations
\be
\mathcal{P}_N(r)\,
=\, 
\sum_{\sigma \in \{+1,-1\}^{N-1}} 
\left(\prod_{n=0}^{N-2} \frac{\lambda_n^{\sigma(n)}}{\lambda_n^+}\right)
\frac{([W_0^{\sigma(0)}]^* e_2) (e_2^* V_{N-2}^{\sigma(N-2)})}{([W_0^{+}]^* e_2) (e_2^* V_{N-2}^{+})}
\cdot \left(\prod_{n=0}^{N-1} \frac{[W_{n+1}^{\sigma(n+1)}]^* V_n^{\sigma(n)}}
{[W_{n+1}^{+}]^* V_n^{+}}\right)\, .
\ee
We know that we are trying to find that the leading order behavior of $\mathcal{R}_N^{(1)}(r)$ is as follows:
it is constant, equal to $\pi^{-1}$, for $r<1$, and it is exponentially small for $r>1$.
We will not try to recover the boundary behavior near $r=1$ in this note.
(But in fact, what we hope to be able to do in a later paper is to calculate $\mathcal{O}_N^{(2)}(z_1,z_2)$
in a similar way, and especially to determine the edge behavior when $z_1$ and $z_2$ are near the circle.)
Let us quickly note how we may dispense with the $r>1$ case so that we may focus on $r<1$.

The largest contribution to $\mathcal{M}_N(r)$ comes from the product of eigenvalues
\be
\begin{split}
\left(\prod_{n=0}^{N-2} \lambda_n^{+}\right)
&=\, \exp\left[\sum_{n=0}^{N-2} \ln( \lambda_n^{+})\right]\\
&=\,e^{(N-1) \ln(N)} \exp\left(\sum_{n=0}^{N-2} \ln\left[ 
\frac{1}{2}\left(r^2+\frac{n+1}{N} + \sqrt{\left(\frac{n+1}{N}-r^2\right)^2+4r^2N^{-1}}\right)
\right]
\right)\, .
\end{split}
\ee
Moreover, defining $t_{n+1} = (n+1)/N$, the sum is $(N-1)$ times a Riemann sum approximation so that:
\be
\begin{split}
\frac{1}{N}\ln\left[e^{-(N-1)\ln(N)} \left(\prod_{n=0}^{N-2} \lambda_n^{+}\right)\right]\, 
&=\,
\int_0^1 \ln\left(\frac{1}{2}\left(r^2+t+\sqrt{(t-r^2)^2+4r^2N^{-1}}\right)\right)\, dt + o(1)\\
&=\, \int_0^1 \ln\left(\frac{1}{2}\left(r^2+t+\sqrt{(t-r^2)^2}\right)\right)\, dt + o(1)\\
&=\, \int_0^1 \ln(\max\{r^2,t\})\, dt + o(1)\, ,
\end{split}
\ee
where the remainder term $o(1)$ is a quantity which converges to $0$ as $N \to \infty$.
Hence we may see, by integrating, that
\be
\lim_{N \to \infty} \frac{1}{N}\ln\left[e^{-(N-1)\ln(N)} \left(\prod_{n=0}^{N-2} \lambda_n^{+}\right)\right]\, 
=\, \begin{cases} \ln(r^2) & \text{ if $r\geq 1$,}\\
r^2-1 & \text{ if $r \in [0,1]$.}
\end{cases}
\ee
But this means that, incorporating the exponential part of the prefactor for $\mathcal{R}_N^{(1)}(r)$,
\be
\lim_{N \to \infty} \frac{1}{N}\ln\left[e^{-(N-1)\ln(N)+N(1-r^2)} \left(\prod_{n=0}^{N-2} \lambda_n^{+}\right)\right]\, 
=\, \begin{cases} 0 & \text { if $r \in [0,1]$,}\\
\ln(r^2)-1+r^2 & \text{ if $r>1$,}
\end{cases}
\ee
and it is easy to see that $\ln(x) \leq x-1$ for all $x \in (0,\infty)$ by convexity of $-\ln(x)$.
So for $r>1$ this is exponentially small: to leading order $e^{N[\ln(r^2)-1+r^2]}$.
We claim that no other factor is exponentially large, so that we obtain
\be
\lim_{N \to \infty} N^{-1} \ln(\mathcal{R}_N^{(1)}(r))\, 
=\, \begin{cases} 0 & \text { if $r \in [0,1]$,}\\
\ln(r^2)-1+r^2 & \text{ if $r>1$.}
\end{cases}
\ee
Therefore, we will henceforth assume $r<1$.

When $r<1$, we claim that we need to do a more careful analysis of the product.
The time scale $t_n=n/N$ is too rough when $t_n$ is near $r^2$.
The purely discrete scale $n$ is too fine.
Therefore, we use the intermediate time scale $T_n = (t_n-r^2)N^{1/2}$, instead.
Then we may rewrite
\be
\lambda_n^{+}\,
=\, N \exp(\psi^{+}(T_{n+1}))\, ,\qquad
\psi^{+}(T_{n+1})\, =\, \ln\left(r^2 + \frac{1}{2}\, N^{-1/2} T_{n+1} + \frac{1}{2}\, N^{-1/2} \sqrt{T_{n+1}^2 + 4 r^2}\right)
\ee
so that
\be
\left(\prod_{n=0}^{N-2} \lambda_n^{+}\right)\,
=\, N^{N-1} \exp\left(\sum_{n=0}^{N-1} \psi^+(T_{n+1})\right)\, .
\ee
Then we  use the Euler-Maclaurin summation
formula to obtain all other terms in the asymptotic series which are significant, including some boundary terms that come with the Euler-Maclaurin formula.
(To do an integral such that $\int_{-r^2 N^{1/2}}^{(1-r^2)N^{1/2}} \psi^+(T)\, dT$, one may find it useful to define $T=2r\sinh(x)$ so that
$dT = 2r \cosh(x)\, dx$ and $\sqrt{T^2+4r^2} = 2r \cosh(x)$, as well.)
Doing all this leads to
\be
\left(\prod_{n=0}^{N-2} \lambda_n^{+}\right)\, \sim\, e \cdot r (1-r^2) e^{N\ln(N) - N(1-r^2)}\, .
\ee
It is also easy to use the definitions of $V_n^+$ and $W_n^+$ to show that
\be
[W_0^{+}]^* e_2\,
\sim\, N^{-1} r^{-2}\quad \text{ and } \quad
e_2^* V_{N-2}^+\, \sim\, N\, .
\ee
Using the Euler-Maclaurin summation formula, one may also prove that
\be
\prod_{n=0}^{N-1} [W_{n+1}^{+}]^* V_n^{+}\,
\sim\,  \frac{r}{1-r^2}\, N^{-1/2}\, .
\ee
The details of the Euler-Maclaurin summation formula for this product as well as for the product of the eigenvalues
are not trivial. (The product of the eigenvalues is harder than the product of the inner-products.)
But they may be done, in particular, by using the intermediate time-scale parameter $T_n$.
Therefore, we obtain
\be
\mathcal{M}_N(r)\, =\, e \sqrt{N}\, e^{(N-1)\ln(N) - N(1-r^2)}\, .
\ee
Therefore, since $\mathcal{R}_N^{(1)}(r) \sim \pi^{-1} (2\pi N)^{-1/2} \exp(-(N-1)\ln(N)+N(1-r^2)) \mathcal{M}_N(r) \mathcal{P}_N(r)$, we see that
\be
\mathcal{R}_N^{(1)}(r)\, \sim\, \pi^{-1}\, \frac{e}{\sqrt{2\pi}}\,   \mathcal{P}_N(r)\, .
\ee
Now we will argue that $\mathcal{P}_N(r)$ is actually independent of $r$, to leading order.

\section{Invariance of the perturbation series $\mathcal{P}_N(r)$}

Let us write
\be
\mathcal{P}_N(r)\, =\, \sum_{\sigma \in \{+1,-1\}^{N-1}} \mathcal{P}_N(\sigma;r)\, ,
\ee
for 
\be
\mathcal{P}_N(\sigma;r)\,
=\, 
\left(\prod_{n=0}^{N-2} \frac{\lambda_n^{\sigma(n)}}{\lambda_n^+}\right)
\frac{([W_0^{\sigma(0)}]^* e_2) (e_2^* V_{N-2}^{\sigma(N-2)})}{([W_0^{+}]^* e_2) (e_2^* V_{N-2}^{+})}
\cdot \left(\prod_{n=0}^{N-1} \frac{[W_{n+1}^{\sigma(n+1)}]^* V_n^{\sigma(n)}}
{[W_{n+1}^{+}]^* V_n^{+}}\right)\, .
\ee
Let us think of $\sigma$ as a sequence of switches, from the $+$ state to the $-$ state, or vice-versa.

Using the notation $t_n = n/N$ and $T_n = N^{1/2} (t_n-r^2)$, we may write
\be
\begin{split}
[W_{n+1}^{\tau}]^* (V_n^{\sigma}-V_{n+1}^{\sigma})\,
&=\, -\frac{\tau N^{-1/2}}{2 \sqrt{T_{n+1}^2+4r^2}}
\left(1+\sigma\, \frac{T_{n}+T_{n+1}}{\sqrt{T_{n}^2+4r^2}+\sqrt{T_{n+1}^2+4r^2}}\right)\\
&=\, -\frac{\tau}{2 \sqrt{T_{n+1}^2+4r^2}}
\left(1+\sigma\, \frac{T_{n}+T_{n+1}}{\sqrt{T_{n}^2+4r^2}+\sqrt{T_{n+1}^2+4r^2}}\right)\, \Delta T_n\, ,
\end{split}
\ee
where we define $\Delta T_n = T_{n+1}-T_n=N^{-1/2}$.
This means that in a time $\Delta T_n$ there is a factor proportional to $\Delta T_n$ contributing to $\mathcal{P}_N(\sigma;r)$,
if we switch from $\sigma=+$ to $\tau=-$ or from $\sigma=-$ to $\tau=+$ because in these cases $[W_{n+1}^{\tau}]^*V_{n+1}^{\sigma}=0$.
This is representative of a Poisson process of jumps.

Moreover, if at $a$ one jumps from $+$ to $-$ and at $b$ one jumps back to $+$, then for all $n \in \{a,\dots,b-1\}$
there is a contribution to $\mathcal{P}_N(\sigma;r)$
equal to
\be
\begin{split}
\frac{\lambda_n^{-}}{\lambda_n^+}
\cdot \frac{[W_{n+1}^{-}]^* V_n^{-}}
{[W_{n+1}^{+}]^* V_n^{+}}\, 
&=\, \frac{r^2+\frac{1}{2}N^{-1/2}T_{n+1}-\frac{1}{2}N^{-1/2}\sqrt{T_{n+1}^2+4r^2}}
{r^2+\frac{1}{2}N^{-1/2}T_{n+1}+\frac{1}{2}N^{-1/2}\sqrt{T_{n+1}^2+4r^2}}\\
&\qquad \qquad
\cdot \frac{1 + \frac{1}{2\sqrt{T_{n+1}^2+4r^2}}
\left(1- \frac{T_{n}+T_{n+1}}{\sqrt{T_{n}^2+4r^2}+\sqrt{T_{n+1}^2+4r^2}}\right)\, \Delta T_n}
{1 - \frac{1}{2\sqrt{T_{n+1}^2+4r^2}}
\left(1+ \frac{T_{n}+T_{n+1}}{\sqrt{T_{n}^2+4r^2}+\sqrt{T_{n+1}^2+4r^2}}\right)\, \Delta T_n}\, ,
\end{split}
\ee
and this quantity is asymptotic to 
$\exp\left(-\left[\frac{\sqrt{T_{n+1}^2+4r^2}}{r^2+\frac{1}{2}N^{-1/2}T_{n+1}}
-\frac{1}{\sqrt{T_{n+1}^2+4r^2}}\right] \Delta T_n\right)$,
when one takes $N \to \infty$ if one also takes a sequence of $T_{n_N}$ such that $|T_{n_N}|/N^{1/2} \to 0$.
Moreover the product is decreasing very rapidly as $|T_n|$ gets large on an order-1 scale.
Therefore, the correction to this asymptotic formula is neglible, for the purpose of calculating the leading order behavior of $\mathcal{P}_N(r)$.
Therefore, defining $\mathcal{P}^{++}_N(r)$ to be the sum of those $\mathcal{P}_N(\sigma;r)$ with $\sigma$ starting at $+$ at the left endpoint
and returning to $+$ at the right endpoint, with some number of intervals of $-$ in between, we have the effect of switching from $+$ to $-$,
staying at $-$ for an interval, and then switching back.
This gives
\be
\begin{split}
\lim_{N \to \infty} \mathcal{P}^{++}_N(r)\,
&=\, 1 + \sum_{K=1}^{\infty} \int_{-\infty<S_1<\dots<S_{2K}<\infty}
\prod_{k=1}^{K} \left[\frac{1}{2\sqrt{S_{2k-1}^2+4r^2}}\left(1+\frac{S_{2k-1}}{\sqrt{S_{2k-1}^2+4r^2}}\right)\right]\\
&\hspace{3cm}
\exp\left(-\sum_{k=1}^{K}\int_{S_{2k-1}}^{S_{2k}} \left[\frac{\sqrt{s^2+4r^2}}{r^2} - \frac{1}{\sqrt{s^2+4r^2}}\right]\, ds\right)\\
&\hspace{3cm}
\prod_{k=1}^{K} \left[-\frac{1}{2\sqrt{S_{2k}^2+4r^2}}\left(1-\frac{S_{2k}}{\sqrt{S_{2k}^2+4r^2}}\right)\right]\, 
dS_1\, \cdots\, dS_{2n}\\
&=\, 1 + \sum_{K=1}^{\infty} (-1)^K \int_{-\infty<x_1<\dots<x_{2K}<\infty}
\prod_{k=1}^{K} \left(\frac{1}{[1+\exp(-2x_{2k-1})][1+\exp(2x_{2k})]}\right)\\
&\hspace{3cm}
\exp\left(-\sum_{k=1}^{K} \int_{x_{2k-1}}^{x_{2k}} [4 \cosh^2(x)-1]\, dx\right)\, dx_1\, \cdots\, dx_{2K}\, ,
\end{split}
\ee
where we made the change of variables $S_k = 2r \sinh(x_k)$, which is useful, as we have also mentioned before.
Let us comment on where the $r$-dependence went.
In fact the limits of integration for $S_1$ and $S_{2K}$ should be $-r^2N^{1/2}<S_1$ and $S_{2K}<(1-r^2)N^{1/2}$.
But, since the exponentials are negative (and growing in magnitude), the integrand is converging rapidly.
Therefore, we can replace the limits of integration, by allowing integrals over all space, with a correction due to the tails of the integrals which are exponentially small.
Then the substitution we have made from $S_k$ to $x_k$ eliminates the $r$ dependence, entirely.
Finally, we mention that we can do the integral in the exponential to simplify the formula, a bit:
\be
\begin{split}
\lim_{N \to \infty} \mathcal{P}^{++}_N(r)\,
&=\, 1 + \sum_{K=1}^{\infty} (-1)^K \int_{-\infty<x_1<\dots<x_{2K}<\infty}
\exp\left(\sum_{k=1}^{K} \left[
-\ln\left(1+e^{-2x_{2k-1}}\right)+\sinh(2x_{2k-1})-x_{2k-1}\right]\right)
\\
&\hspace{3cm}
\exp\left(-\sum_{k=1}^{K} \left[
\ln\left(1+e^{2x_{2k}}\right)+\sinh(2x_{2k})-x_{2k}\right]\right)
\, dx_1\, \cdots\, dx_{2K}\\
&\hspace{-2cm}=\, 1 + \sum_{K=1}^{\infty} (-1)^K \int_{-\infty<x_1<\dots<x_{2K}<\infty}
e^{-\sum_{k=1}^{K} \left(\ln[\cosh(x_{2k-1})]+\ln[\cosh(x_{2k})]+\sinh(2x_{2k})-\sinh(2x_{2k-1})\right)}
\, dx_1\, \cdots\, dx_{2K}\, .
\end{split}
\ee
Again, note that this is rapidly decreasing as $x_1 \to -\infty$ or $x_{2K} \to \infty$.
But to get the analogous terms $\mathcal{P}^{+-}_N(r)$, $\mathcal{P}^{-+}_N(r)$
and $\mathcal{P}^{--}_N(r)$, we can just alter this formula essentially by taking $x_1 \to -\infty$
or $x_{2K} \to \infty$ or both.
(This is not entirely correct because we lose terms corresponding to the density for crossing. But morally it is still correct
because the terms remaining are certainly going to $0$.)
Therefore 
\be
\lim_{N \to \infty} \mathcal{P}_N(r)\,
=\, \lim_{N \to \infty} \mathcal{P}_N^{++}(r)\, .
\ee
Since we know that $\lim_{N \to \infty} \mathcal{R}_N^{(1)}(r)$ must equal $\pi^{-1}$ on the disk (for instance because the area of the disk is 1)
this leaves the calculation to show that
\be
1 + \sum_{K=1}^{\infty} (-1)^K \int_{-\infty<x_1<\dots<x_{2K}<\infty}
e^{-\sum_{k=1}^{K} \left(\ln[\cosh(x_{2k-1})]+\ln[\cosh(x_{2k})]+\sinh(2x_{2k})-\sinh(2x_{2k-1})\right)}
\, dx_1\, \cdots\, dx_{2K}\,
\stackrel{?}{=}\, \frac{\sqrt{2\pi}}{e}\, .
\ee
At this time we cannot see a direct method to prove this.
But we hope to explore it in a later paper.

\section{Summary and outlook}

We have considered the complex Ginibre ensemble.
We consider the problem of calculating the mixed matrix moments to be a nice pedagogical problem.
It may be used to illustrate the method of using concentration of measure to derive nonlinear recursion relations.
This method is particularly important in spin glass theory, where it led to the Ghirlanda-Guerra identities, which are critical
to those models.

The most natural connection between spin glasses and random matrices are the spherical spin glasses
of \cite{KTJ} and \cite{CS}.
This has been studied vigorously with very detailed results.
See for example \cite{ABAC}.
The relation we have drawn between the overlaps in spin glasses and the moments in random
matrix theory is mainly illustrative, to suggest the central role of concentration-of-measure (COM).
In addition to spin glass theory and random matrix theory, the idea of using COM to derive low-dimensional nonlinear equations
to replace linear equations in high dimensions is helpful in a variety of contexts \cite{CK}.

The mixed matrix moments for the complex Ginibre ensemble are particularly nice moments to consider because their combinatorics
is as simple as possible. (Indeed it is somewhat simpler than the usual Catalan numbers that arise in the GUE/GOE moments
or the bipartite Catalan numbers that arise in the Mar\v{c}enko-Pastur law.)
Also, they are not as well-studied as the other moments for the classical Gaussian matrix ensembles.
But they are still well-studied.
However, an interesting facet which has not been exhaustively studied is their relation to the overlap functions defined by Chalker and Mehlig.

Chalker and Mehlig's papers are extremely interesting and introduce what certainly seems like a key 
object in random matrix theory that has not been taken up sufficiently yet by mathematicians.
It is recognized as a key result by theoretical and mathematical physicists.
See, for instance, the recent paper \cite{Burda}.
One interesting question is how the free probability theory approach to sums of independent
random matrices is altered for non-Hermitian matrices, and the extent to which it relates to 
the eigenvector overlap kernels.
This has initially been considered in \cite{Jarosz}.

Chalker and Mehlig did not consider the application of calculating the mixed matrix moments from their overlap functions.
Indeed, since the mixed matrix moments are already known, the reverse problem seems more reasonable.
But it would probably be very difficult to calculate the overlap functions just from the mixed matrix moments.
However, what is true is that, if one takes Chalker and Mehlig's formula for the bulk overlap functions,
then the mixed matrix moments do place some constraints on the edge behavior,
as we have shown. 

We have proposed a possible method for calculating $\mathcal{O}_N^{(2)}(z_1,z_2)$, asymptotically.
But we have not carried out this suggestion.
We did illustrate it by re-deriving $\mathcal{R}_N^{(1)}(z)$ by treating the second-order recursion formula
as an adiabatic matrix evolution problem.

Now we would like to suggest another interesting direction for further study.
Fyodorov and Mehlig, and Fyodorov and Sommers, calculated two very interesting examples
of non-Hermitian random matrices for which they obtained exact expressions for the overlap functions \cite{FY1,FY2}.
They did not yet calculate the mixed matrix moments for these random variables.
It would be an ideal problem to do so, and check the formulas linking the overlap functions and the mixed matrix moments.

In a private communication, Fyodorov has explained that the eigenfunction non-orthogonality
in the systems considered in \cite{FY1,FY2} has physical relevance.
The overlap was shown by Fyodorov and Savin to give the resonance shift if one perturbs a scattering
system \cite{FY3}.
This was even experimentally verified recently \cite{Gros}.

Finally, the first two overlap functions only help with calculating mixed matrix moments of the Ginibre ensemble of the 
form $\operatorname{tr}[A^p (A^*)^p]$ for $p=1,2,\dots$.
In order to calculate mixed matrix moments for more than two factors one needs higher order overlap functions.
Given the difficulty to calculate the first two, this is a formidable problem.
But it might be a reasonable exact calculation for the matrix ensembles considered by Fyodorov and his collaborators.

\subsection{Open Questions related to the physical interpretation}

1. In \cite{FY3}, Fyodorov and Savin demonstrate the relevance of non-Hermitian matrices for the important problem of modelling resonances in non-equilibrium models of quantum statistical mechanics.
The lifetime associated to a quasi-bound state is proportional to the inverse of the resonance width $\Gamma$, which is the positive imaginary part of the eigenvalue.
It is therefore a canonical mathematical question to investigate the greatest lifetime, i.e., the smallest resonance width in any given model.
In the formalism of \cite{FY3} this has been done. It seems like an important question to also try to do it at perfect coupling, as well as weak coupling.

2. Motivated by problem 1, a canonical question for the complex Ginibre ensemble is to characterize the largest real part of any eigenvalue.
The largest amplitude was discovered by Ginibre, and it follows a Gumbel distribution for the fluctuations.
In fact, even in the real Ginibre ensemble, a similar result was obtained by \cite{Rider}.
But in the context of resonances, the largest real or imaginary part is a separate question of interest.

\section*{Acknowledgments}
We are very grateful to Y.~Fyodorov for useful suggestions and for references to the literature.
We are also grateful to an anonymous referee for various helpful suggestions.
S.~S.\ gave a talk on this topic at the Banff workshop,``Spin Glasses and Related Topics,''
and some of the work was completed there.
He is very grateful to the center and the organizers.


\begin{thebibliography}{10}

\bibitem{AC}
M.~Aizenman and P.~Contucci.
\newblock On the Stability of the Quenched state in Mean Field
Spin Glass Models.
\newblock {\em J.~Statist.~Phys.} {\bf 92}, 765--783, (1998).

\bibitem{AndersonGuionnetZeitouni}
Greg W.~Anderson, Alice Guionnet and Ofer Zeitouni.
\newblock {\em An Introduction to Random Matrices.}
\newblock Cambridge University Press, 2009.

\bibitem{ABAC}
Antonio Auffinger, G\'erard Ben Arous and Ji\v{r}\'i \v{C}ern\'y.
\newblock Random Matrices and Complexity of Spin Glasses.
\newblock {\em Comm.~Pure Appl.~Math.} {\bf 66}, no.~2, 165--201 (2013).


\bibitem{Burda}
Z.~Burda, J.~Grela, M.~A.~Nowak, W.~Tarnowski and P.~Warcho\l.
\newblock Dysonian dynamics of the Ginibre ensemble.
\newblock {\em Preprint}, 2014.
\newblock \url{http://arxiv.org/abs/1403.7738}

\bibitem{CM}
J.T.~Chalker and B.~Mehlig.
\newblock Eigenvector Statistics in Non-Hermitian Random Matrix Ensembles.
\newblock {\em Phys.~Rev.~Lett.} {\bf 81}, 3367 (1998).

\bibitem{CK}
S.~Chatterjee and K.~Kirkpatrick.
\newblock Probabilistic Methods for Discrete Nonlinear Schr\"odinger Equations.
\newblock  {\em Commun.~Math.~Phys.}, {\bf 65}, no.~5, 727--757 (2012).

\bibitem{ContucciGiardina}
Pierluigi Contucci and Cristian Giardin\`a.
\newblock The Ghirlanda-Guerra identities.
\newblock {\em J.~Statist.~Phys.} {\bf 126}, no.~4, 917--931, (2007).

\bibitem{CS}
A.~Crisanti and H-J.~Sommers.
\newblock Thouless-Anderson-Palmer approach to the spherical p-spin
spin glass model.
\newblock {\em J.~Phys.~I France} {\bf 5} no. 7, 805 (1995).

\bibitem{Forrester}
P.~J.~Forrester and G.~Honner.
\newblock Exact statistical properties of the zeros of complex random polynomials.
\newblock {\em J.~Phys.~A: Math.~Gen.} {\bf 32}, 2961 (1999).

\bibitem{FY1}
Yan V.~Fyodorov and B.~Mehlig.
\newblock Statistics of resonances and nonorthogonal eigenfunctions in a model
for single-channel chaotic scattering.
\newblock {\em Phys.~Rev.~E} {\bf 66}, 045202(R) (2002).

\bibitem{FY3}
Yan V.~Fyodorov and Dmitry V.~Savin.
\newblock Statistics of Resonance Width Shifts as a Signature of Eigenfunction Nonorthogonality.
\newblock {\em Phys.~Rev.~Lett.} {\bf 108} 184101 (2012).

\bibitem{FY2}
Yan V.~Fyodorov and H-J Sommers.
\newblock Random matrices close to Hermitian or unitary:
overview of methods and results.
\newblock {\em J.~Phys.~A: Math.~Gen.} {\bf 36} 3303--3347 (2003).

\bibitem{GG}
S.~Ghirlanda and F.~Guerra.
\newblock General properties of overlap probability distributions
in disordered spin systems. Towards Parisi ultrametricity.
\newblock {\em J.~Phys.~A: Math.~Gen.} {\bf 31}, 9149--9155 (1998).

\bibitem{Gros}
J.-B.~Gros, U.~Kuhl, O.~Legrand, F.~Mortessagne, E.~Richalot and D.~V.~Savin.
\newblock Experimental width shift distribution: a test of nonorthogonality for local and global perturbations.
\newblock {\em Preprint,} 2014.
\newblock \url{http://arxiv.org/abs/1408.6472}

\bibitem{G}
F.~Guerra.
\newblock Broken replica symmetry bounds in the mean field spin glass model.
\newblock {\em Commun.~Math.~Phys.}, {\bf 233}(1):1--12, (2003).

\bibitem{GT}
F.~Guerra and F.L.~Toninelli.
\newblock The thermodynamic limit in mean field spin glass models.
\newblock {\em Commun.~Math.~Phys.}, {\bf 230}:71--79, (2002).

\bibitem{Janick}
R.~A.~Janik, W.~N\"orenberg, M.~A.~Nowak, G.~Papp, and I.~Zahed.
\newblock Correlations of eigenvectors for non-Hermitian random-matrix models.
\newblock {\em Phys.~Rev.~E} {\bf 60} (1999) 2699.

\bibitem{Jarosz}
A.~Jarosz and M.~A.~Nowak.
\newblock Random Hermitian versus random non-Hermitian
operators--unexpected links.
\newblock {\em J.~Phys.~A: Math.~Gen.} {\bf 39} (2006) 10107--10122.

\bibitem{Kanzieper}
E.~Kanzieper.
\newblock Exact replica treatment of non-Hermitean complex random matrices.
\newblock  In: {\em Frontiers in Field Theory,} edited by O. Kovras, Ch. 3, pp. 23 -- 51 (Nova Science Publishers, NY 2005). 

\bibitem{Kemp}
Todd Kemp , Karl Mahlburg, Amarpreet Rattan and Clifford Smyth.
\newblock Enumeration of non-crossing pairings on bit strings.
\newblock
 {\em J.~Combin.~Theory Ser.~A} {\bf 118}, 129--151 (2011).

\bibitem{KTJ}
J.~M.~Kosterlitz, D.~J.~Thouless, and Raymund C.~Jones.
\newblock Spherical Model of a Spin-Glass.
\newblock {\em Phys.~Rev.~Lett.} {\bf 36}, 1217 (1976).

\bibitem{MC}
B.~Mehlig and J.T.~Chalker.
\newblock Statistical properties of eigenvectors in non-Hermitian Gaussian random matrix
ensembles.
\newblock {\em J.~Math.~Phys.} {\bf 41}, 3233 (2000).

\bibitem{Mehta}
Madan Lal Mehta.
\newblock {\em Random Matrices. Third Edition.}
\newblock Elsevier, 2014.

\bibitem{Panchenko}
Dmitry Panchenko
\newblock The Parisi ultrametricity conjecture.
\newblock {Annals of Mathematics} {\bf 177}, Issue 1, 383--393, (2013).

\bibitem{Pastur}
L.~A.~Pastur.
\newblock Spectra of random self adjoint operators.
\newblock {\em Russ.~Math.~Surv.} {\bf 28}, n.~1 (1973).

%\bibitem{Prolhac}
%S.~Prolhac.
%\newblock Spectrum of the totally asymmetric simple exclusion process on a periodic lattice-first excited states.
%\newblock {\em J.~Phys.~A: Math.~Theor.} {\bf 47} 375001.

\bibitem{Rider}
B.~Rider and C.~D.~Sinclair.
\newblock Extremal laws for the real Ginibre ensemble.
\newblock {\em Ann.~Appl.~Probab.} {\bf 24}, no.~4, 1621--1651 (2014).

\bibitem{Talagrand}
Michel Talagrand.
\newblock{\em Mean Field Models for Spin Glasses: Volume I: Basic Examples.}
\newblock Springer-Verlag, 2010.

\end{thebibliography}
\end{document}